\documentclass[twoside,11pt]{article}

\usepackage{blindtext}

\usepackage[preprint]{jmlr2e}
\usepackage{amsmath, booktabs, multicol, multirow, enumerate}

\def\tr{\text{\rm tr}}
\def\vec{\text{\rm vec}}

\def\Cov{\text{\rm Cov}}
\def\Var{\text{\rm Var}}

\def\KL{\text{\rm KL}}

\def\N{\text{N}}
\def\argmin{\text{\rm argmin}}
\def\diag{\text{\rm diag}}

\def\vech{\text{\rm vech}}

\def\bfd{{\bf d}}
\def\bfg{{\bf g}}

\def\bfn{{\bf n}}

\def\bfx{{\bf x}}
\def\bfy{{\bf y}}

\def\bfw{{\bf w}}
\def\bfu{{\bf u}}

\def\bfpi{{\pmb \pi}}

\def\dt{{\mathrm{d}}}
\newcommand{\wh}{\widehat}
\def\Mn{{\mathsf{M}}} 
\def\Ln{{\mathsf{L}}}

\def\cut{{\mathrm{cut}}}
\def\ls{{\text{ls}}}
\def\dv{{\text{dv}}}

\def\E{{\mathbb{E}}}

\def\PG{{\text{PG}}}

\def\P{\text{P}}

\newcommand\mL{{\mathcal{L}}}
\def\mcF{{\mathcal{F}}}

\def\Mn{{\mathsf{M}}} 
\def\Ln{{\mathsf{L}}} 

\newtheorem{assumption}{Assumption} 
\newtheorem{cor}{Corollary}
\newtheorem{rem}{Remark}

\newtheorem{nonitalicthm}{Theorem}

\usepackage{lastpage}
\jmlrheading{0}{2026}{1-\pageref{LastPage}}{1/26; Revised 0/26}{0/26}{21-0000}{Tan, Nott and Frazier}

\ShortHeadings{Optimization-centric cutting feedback}{Tan, Nott and Frazier}
\firstpageno{1}
\graphicspath{{figures/}}
\allowdisplaybreaks 

\begin{document}

\title{Optimization-centric Cutting Feedback \\ for Semiparametric Models}

\author{\name Linda S. L. Tan \email statsll@nus.edu.sg \\
       \addr Department of Statistics and Data Science \\
       National University of Singapore\\
       Singapore 117546, Singapore
       \AND
       \name David J. Nott \email standj@nus.edu.sg \\
       \addr Department of Statistics and Data Science \\
       National University of Singapore\\
       Singapore 117546, Singapore
	   \AND
	   \name David T. Frazier \email david.frazier@monash.edu \\
	   \addr Department of Econometrics and Business Statistics \\
	   Monash University \\
	   Victoria 3800, Australia}

\editor{My editor}

\maketitle

\begin{abstract}
Complex statistical models are often built by combining multiple submodels, called modules. Here we consider modular inference where the modules contain both parametric and nonparametric components. In such cases, standard Bayesian inference can be highly sensitive to misspecification in any module, and influential prior specifications for the nonparametric components can compromise inference for the parametric components, and vice versa. We propose a novel ``optimization-centric'' approach to cutting feedback for semiparametric modular inference, which can address misspecification and prior-data conflicts. The proposed cut posteriors are defined via a variational optimization problem like other generalized posteriors, but regularization is based on  R\'{e}nyi divergence, instead of Kullback-Leibler divergence (KLD). We show empirically that defining the cut posterior using R\'{e}nyi divergence delivers more robust inference than KLD, and R\'{e}nyi divergence reduces the tendency to underestimate uncertainty when the variational approximations impose strong parametric or independence assumptions. Novel posterior concentration results that accommodate the R\'{e}nyi divergence and allow for semiparametric components are derived, extending existing results for cut posteriors that only apply to KLD and parametric models. These new methods are demonstrated in a benchmark example and two real examples:  Gaussian process adjustments for confounding in causal inference and misspecified copula models with nonparametric marginals.
\end{abstract}


\section{Introduction}\label{sec-intro}

Complex Bayesian models are often specified by joining together submodels for different data sources, called modules. However, if some modules are misspecified, conventional Bayesian inference becomes unreliable. {\em Cutting feedback} \citep{Plummer2015, Jacob2017} modifies standard Bayesian inference so that a misspecified module does not contaminate inference for parameters in correctly specified modules. Recently, \cite{Frazier2024} developed cutting feedback methods in a generalized Bayesian framework \citep{Bissiri2016} for parametric models, where the negative log-likelihoods in misspecified modules are replaced by general loss functions, but retained for well-specified modules. This provides a flexible approach to model expansion when some module likelihoods are misspecified.

The ability of cutting feedback methods to mitigate the impact of misspecification in parametric Bayesian models has led to their application across diverse fields, such as pharmacokinetic and pharmacodynamic models \citep{Bennett2001, Zhang2003a, Zhang2003b}, air pollution studies \citep{Blangiardo2011}, causal inference \citep{Zigler2013} and archaeological studies \citep{Styring2017}. Cutting feedback was first implemented using modified Gibbs or Metropolis-within-Gibbs posterior sampling, where some model terms are omitted when forming conditional distributions in MCMC, to mitigate the effects of misspecified model terms on inference for some parameters. However, such implicit definitions of cut posteriors make it difficult to understand the effects of cuts. \cite{Plummer2015} studied a two-module system, which defines cutting feedback explicitly and can be used in many important applications. \cite{Yu2023} formulate the joint cut posterior in this two-module system as the solution to a constrained variational optimization problem, and \cite{Smith2025} extend this approach to misspecified copula models. \cite{Song2025} propose variational inference based on normalizing flows for the second module that relies only on samples from the first module to reduce approximation error. \cite{Liu2025} extend the two-module system of \cite{Plummer2015} to an arbitrary number of modules. Partial cuts are considered in semi-modular inference \citep{Carmona2020}, where a degree of influence parameter interpolates between the cut and conventional posterior, and are optimal in mildly misspecified models \citep{chakraborty2023modularized, frazier2023accurate}.

In this article we study modular models that contain both parametric and nonparametric components, in contrast to much of the existing literature. Modular semiparametric models have been used in several applications, such as computer model calibration \citep{Bayarri2009}, semiparametric hidden Markov models \citep{Moss2024} and localized geographic regression \citep{Liu2024}. However, cutting feedback with nonparametric components presents two challenges that must be overcome to ensure effectiveness. First, in modules with both parametric and nonparametric components, prior specification is crucial and priors that deliver reliable inference on one component in one module may deliver poor inferences for parameters in other modules. For instance, in hidden Markov models, \cite{Moss2024} exhibit a scenario where useful priors for parametric components deliver poor inferences on nonparametric components, while the reverse situation is a well-known issue in semiparametric Bayesian models \citep{rivoirard2012posterior}.  Second, the same influential prior choices can negatively affect inferences within modules even if the likelihood is correctly specified, so that it is necessary to diagnose and control the effects of any prior-data conflict \citep{evans2006}. We show that both issues can be overcome by building on \cite{Knoblauch2022} to develop an \textit{optimization-centric cut posterior} (OCCP).

Our work makes three main contributions. First, we extend the framework of \cite{Frazier2024} by proposing an optimization-centric generalized Bayes formulation of cutting feedback. We show that any cut posterior in a two-module system can be formulated as the solution to a variational optimization that depends on the choice of: (1) the loss for each module; (2) the learning rates that balance information in the loss relative to the prior; and (3) the divergence measure between prior and posterior that includes Kullback-Leibler divergence (KLD) as a special case. Our OCCP unites the definition of \cite{Plummer2015} with the variational formulation of likelihood-based cut posteriors in \cite{Carmona2021} and \cite{Yu2023}, while extending \citet{Frazier2024} to divergences beyond KLD. Our second main contribution is to show that OCCPs can mitigate the risks of inappropriate prior specification and computational approximations in complex models. Our empirical results indicate that R\'{e}nyi divergence based OCCPs can reduce the impact of prior-data conflicts, and yield variational approximations that quantify uncertainty more accurately than those based on KLD. We apply the OCCPs to several parametric and semiparametric examples, including Gaussian process (GP) adjustments for confounding in causal inference, and nonparametric marginal estimation in misspecified copula models. Finally, we present theoretical results that greatly extend existing results for general cut posteriors, by allowing for semiparametric settings where some modules may contain infinite dimensional parameters, considering the broader class of R\'{e}nyi divergence beyond KLD, and developing results for variational approximations of the OCCPs.

To motivate our approach, it is helpful to consider a biased means example discussed in \cite{Bayarri2009}. Although this example is parametric, we consider semiparametric variants of it later in Section 5 involving functional parameters. In the example of \cite{Bayarri2009} there is a small sample $z=(z_1,\dots, z_{n_1})$ from a reliable data source with mean $\varphi$, and a large sample $w=(w_1,\dots w_{n_2})$ from a biased data source with mean $\varphi+\eta$, where $\eta$ is the bias. Interest centers on the mean parameter $\varphi$, and the dilemma is whether to use the large sample of biased data $w$ to supplement information from $z$.  \cite{Bayarri2009} show that this is detrimental when the prior for $\eta$ is misspecified, and is not beneficial even if the prior is well-specified. Cutting feedback is more attractive than standard Bayesian inference in this example. In Section \ref{sec biased mean}, we consider misspecified priors not only for $\eta$ but $\varphi$ as well. In this case, conventional cutting feedback fails, but the OCCPs provide reliable inference. Figure 1 shows, for a simulated data set, the true posterior for $(\varphi,\eta)$, conventional KLD based cut posterior, and an OCCP based on R\'{e}nyi's $\alpha$ divergence. When $\alpha$ is small, the true value is in the high probability region of the OCCP, but both the true and conventional cut posterior provide misleading inference. This example is discussed further in Section \ref{sec biased mean}.  When $\varphi$ or $\eta$ are infinite-dimensional parameters,
the risks that come with influential prior specifications are greater, as well
as the challenges of cut posterior computations, as discussed in the examples
of Section 5.

\begin{figure}[tb!]
\centering
\includegraphics[width=0.8\textwidth]{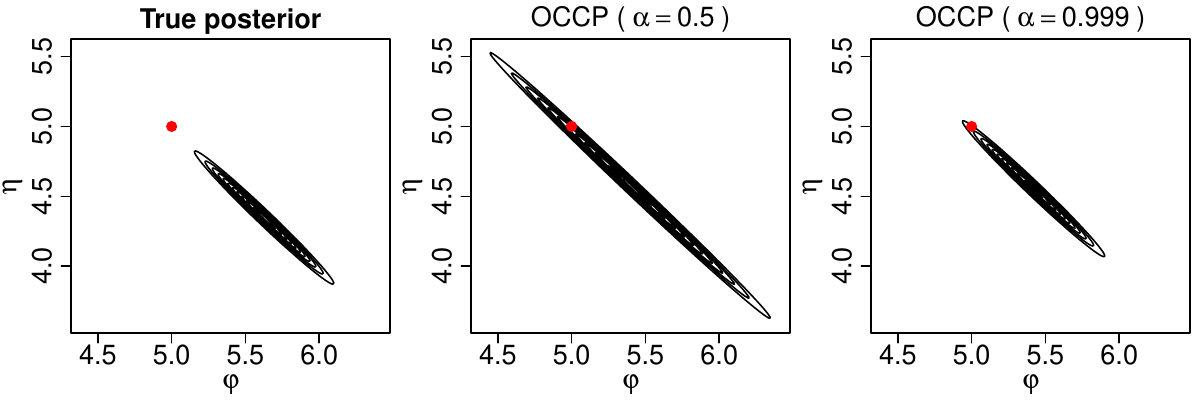}
\caption{\label{Fig1} Contour plots of joint posteriors for $\theta=(\varphi,\eta)$, with priors $\varphi\sim \N(10,1)$ and $\eta\sim \N(0,1)$. True value of $\theta$ is marked in red. Both the true and conventional cut posterior (approximated with $\alpha=0.999$) provide poor inference.}
\end{figure}

In Section \ref{sec cutting feedback}, we define and motivate cutting feedback methods, describe the generalized Bayesian approach to cutting feedback, and extend it to optimization-centric generalized Bayes. Section \ref{sec: examples} considers some examples where OCCPs can be applied to derive inference that is more robust to model misspecification than standard Bayesian inference. Section \ref{sec posterior concentration} presents theoretical posterior concentration results for OCCPs, which apply to infinite-dimensional parameters and semiparametric modular inference. Our approach is illustrated using two examples, on removing hidden confounding in a GP model for an observational study in Section \ref{sec hidden confounding}, and fitting a misspecified copula model with nonparametric marginals in Section \ref{sec misspecified copula model}. Section \ref{sec conclusion} concludes and gives some directions for future work.

\section{Cutting feedback} \label{sec cutting feedback}
Complex models can often be decomposed into submodels or modules. Bayesian inference is usually based on posterior distributions derived from Bayes' Theorem, but in multi-modular settings it may be desirable to modify the conventional posterior by omitting certain terms. Possible reasons include preventing a misspecified module from corrupting a well-specified module, lack of identifiability, modular development and computational complexity \citep{Bayarri2009}. Prevention of information flow from some modules to others is known as {\em cutting feedback}.

Throughout this article, we consider the two-module system of \cite{Plummer2015}, which has modules for two data sources, $z = (z_1, \dots, z_{n_1})^\top$ and $w = (w_1, \dots, w_{n_2})^\top$. While $z$ is connected to a parameter $\varphi \in \Phi$ in the first module, $w$ is connected to $\varphi$ and another parameter $\eta \in \mathcal{E}$ in the second module. Let $\theta=(\varphi^\top, \eta^\top)^\top$ and $p(\theta) = p(\varphi)p(\eta \mid \varphi)$ be the prior. Writing $y=(z^\top,w^\top)^\top$ for the full data, the joint likelihood is $p(y\mid \theta)=p(z \mid \varphi)p(w \mid \theta)$, where $p(z \mid \varphi)$ and $p(w\mid \theta)$ are likelihoods for the first and second module respectively. From Bayes' theorem, the joint posterior is $p(\theta \mid y) = p(\varphi \mid y) p(\eta \mid  \varphi, w)$.

Suppose we wish to prevent information flow from $w$ to $\varphi$ in the posterior, due to computational complexity or to avoid contaminating the inference of $\varphi$ with ``bad'' data $w$. \cite{Plummer2015} defined the {\em cut joint posterior} as
\begin{align}\label{cutpost} 
p_\cut (\theta \mid y) = p(\varphi \mid   z)p(\eta \mid  \varphi, w),
\end{align}
which replaces $p(\varphi \mid y)$ with $p(\varphi \mid z)$ in $p(\theta \mid y)$. Uncertainty about $\varphi$ is still propagated in inference for  $\eta$. As $p(\varphi \mid  y) \propto p(\varphi \mid  z) \int p(w \mid  \theta)  p (\eta \mid \varphi) \, \dt\eta$, the integral representing influence of $w$ on $p(\varphi \mid  y)$ is dropped to yield the cut posterior.

\subsection{Optimization-centric generalized Bayesian inference} \label{Sec OC gen Bayes}
Conventional Bayesian updating can be formulated as the solution of a variational optimization problem \citep{zellner88}. Building on this perspective, \cite{Knoblauch2022} develop a form of generalized Bayesian inference that defines a generalized posterior by modifying the optimization problem defining the usual Bayesian update. This modification can achieve robustness to model misspecification and enhance computational scalability. To explain further, let $y=(y_1,\dots, y_n)^\top$ be the data, $\theta$ be a parameter of interest with prior $p(\theta)$ and $R(y \mid \theta)=\sum_{i=1}^n r_i(y_i|\theta)$ be a loss function.  As \cite{Knoblauch2022} do not consider modular structures, $y$ is not split into module-specific data sources. They consider a class $\mathcal{F}$ of candidate posteriors, and a divergence $\mathcal{D}(q\|p)$ measuring deviation of $q\in \mathcal{F}$ from the prior. The optimization-centric posterior of $\theta$ is defined as
\begin{equation} \label{ocpost1}
p^*(\theta \mid y) = \argmin_{q(\theta)\in\mathcal{F}} \quad
\lambda \, \E_{q(\theta)}\{R(y \mid \theta)\}
+\mathcal{D}(q\|p),  
\end{equation}
where $\lambda>0$ is a learning rate to scale the loss suitably relative to the divergence, with $\lambda=1$ chosen in \cite{Knoblauch2022}. Minimizing \eqref{ocpost1} compromises between minimizing the expected loss and making $q\in\mathcal{F}$ close to the prior. If $\lambda=1$, $R(y\mid\theta)$ is the
negative log-likelihood, $\mathcal{F}$ is unrestricted and $\mathcal{D}(q\|p) = \E_{q(\theta)}\left\{\log q(\theta)-\log p(\theta)\right\}$ is the KLD, then 
\eqref{ocpost1} is the conventional posterior.

Specification of the loss, divergence and admissible posteriors in \eqref{ocpost1} address three issues in standard Bayesian inference. First, if framing a probabilistic model for the data is difficult, we can replace the negative log-likelihood with a loss function that reflects our inferential goals. Second, if the prior has conflict with information about $\theta$ in the loss function, a robust divergence can downweight prior influence. Third, if computation is intractable, we can choose $\mathcal{F}$ so that the optimization problem is easier to solve. When $\mathcal{F}$ is unrestricted and $\mathcal{D}(q\|p)$ is the KLD, the solution to \eqref{ocpost1} is
\begin{align}\label{gibbspost}
p^*(\theta \mid y) & \propto p(\theta)\exp\left\{-\lambda R(y \mid \theta)\right\},  
\end{align}
which is the Gibbs posterior \citep{Bissiri2016} often studied in PAC-Bayesian generalization bounds \citep{Alquier2024}.  Setting $\lambda=1$ and $R(y\mid \theta)$ as the negative log-likelihood recovers the conventional posterior.

In the two-module system with $\theta = (\varphi^\top,\eta^\top)^\top$, \cite{Carmona2020} observed that the cut posterior in \eqref{cutpost} can be interpreted as a Gibbs posterior with 
\[
R(y \mid \theta)= -\log p(z \mid \varphi)-\log p(w \mid \theta) + \log p(w \mid \varphi)
\]
and $\lambda=1$. This is immediate by plugging $R(y \mid \theta)$ into \eqref{gibbspost}. \cite{Frazier2024} extend the concept of cut posterior to a {\em generalized cut posterior} by replacing individual module negative log-likelihoods by loss functions. They employ the KLD, $\lambda=1$, an unrestricted class $\mathcal{F}$, and the loss function
\begin{align*}
R(y \mid \theta) = \lambda_1 L(z\mid \varphi)
+ \lambda_2 M(w \mid \theta) + \log m_\eta(w \mid \varphi),
\end{align*}
where $L(z \mid  \varphi) = \sum_{i=1}^{n_1} \ell_i(z_i \mid \varphi)$, $M(w\mid \theta) = \sum_{j=1}^{n_2} m_j(w_j\mid \theta)$, $\lambda_1,\lambda_2>0$ are learning rates, one for each module-specific loss function, and $m_\eta(w \mid \varphi) =\int p(\eta \mid \varphi)\exp \{-\lambda_2 M(w\mid \theta)\}\, \dt \eta$ is a normalizing constant for $p^*(\eta\mid \varphi,w)$ defined in \eqref{mfn} below. Plugging $R(y\mid\theta)$ into \eqref{gibbspost}, the generalized cut posterior is  
$p^*(\theta\mid y) = p^*(\varphi \mid z)p^*(\eta \mid \varphi,w)$, where $p^*(\varphi\mid z) \propto p(\varphi)\exp\left\{-\lambda_1 L(z\mid \varphi)\right\}$ and 
\begin{align} \label{mfn}
p^*(\eta\mid \varphi,w) = \dfrac{p(\eta \mid \varphi)\exp\left\{-\lambda_2 M(w\mid \theta)\right\}}{m_\eta(w \mid \varphi)}.
\end{align}

\subsection{Optimization-centric cut posterior} \label{Sec OC cut posterior}
In \cite{Frazier2024}, KLD is used to measure the distance of $q\in \mathcal{F}$ from the prior. We now propose an optimization-centric cut posterior (OCCP), 
\[
p_\cut^* (\theta \mid y) = p_\cut^*(\varphi \mid z) 
p^*(\eta \mid \varphi, w) ,
\]
which allows the use of other divergences by directly defining an optimization-centric marginal posterior for $\varphi$ and conditional posterior for $\eta$ given $\varphi$, where
\begin{align} \label{cut_varphi}
p_\cut^*(\varphi \mid z) &= \argmin_{q(\varphi) \in \mcF_\varphi} \;\; \lambda_1 \, \E_{q(\varphi)} \{L(z \mid  \varphi)\} + \mathcal{D}_\varphi\{ q(\varphi) \| p(\varphi) \} , \\
p^*(\eta \mid  \varphi, w) &= \argmin_{ q(\eta \mid \varphi) \in \mcF_{\eta \mid  \varphi} } \; \lambda_2 \, \E_{q(\eta \mid \varphi)} \{ M(w \mid  \theta) \} +  \mathcal{D}_{\eta \mid \varphi} \{ q(\eta \mid  \varphi) \| p(\eta\mid \varphi) \}. \nonumber
\end{align}
In \eqref{cut_varphi}, $\mathcal{D}_\varphi$ and $\mathcal{D}_{\eta \mid \varphi}$ are divergence measures while $\mcF_\varphi$ and $\mcF_{\eta \mid \varphi}$ are classes of densities to optimize over. As in  \cite{Frazier2024}, we introduce two learning rates $\lambda_1,\lambda_2 > 0$ to scale the expected loss relative to the divergence in each module. If loss functions in \eqref{cut_varphi} are negative log-likelihoods for a probabilistic model, KLD is used and $\lambda_1=\lambda_2=1$, then the solutions to \eqref{cut_varphi} will be the marginal and conditional posterior in the cut posterior in \eqref{cutpost}, provided these are contained in $\mcF_\varphi$ and $\mcF_{\eta \mid \varphi}$. We refer to densities in $\mcF_\varphi$ and $\mcF_{\eta\mid \varphi}$ as variational densities and their parameters as variational parameters. Lemma \ref{lem:cutpost} shows that our definition of an OCCP in \eqref{cut_varphi} is equivalent to an optimization-centric joint posterior in \cite{Knoblauch2022}.

\begin{lemma}\label{lem:cutpost} 
Consider a joint generalized posterior density for $\theta$ optimized over the class $\mcF_{\theta}=\{q(\theta) \mid q(\varphi) \in \mcF_{\varphi}, q(\eta\mid \varphi) \in \mcF_{\eta\mid \varphi} \}$, based on the divergence
\[
\mathcal{D}_{\theta} \{q(\theta) \| p(\theta)\} 
= \mathcal{D}_\varphi\{ q(\varphi) \| p(\varphi) \} 
+ \E_{q(\varphi)} [ \mathcal{D}_{\eta\mid \varphi} \{ q(\eta \mid \varphi) \| p(\eta \mid \varphi)\} ],
\]
where $\mathcal{D}_\varphi$ and $\mathcal{D}_{\eta\mid \varphi}$ are the divergences in \eqref{cut_varphi}. Let $\lambda =1$ and consider the loss
\begin{align*} 
R(y\mid \theta) &=\lambda_1 L(z\mid\varphi)+\lambda_2 M(w\mid \theta) + N(w\mid \varphi), \quad \text{where} \\
N(w\mid \varphi) &= -\min_{q(\eta\mid\varphi)\in \mcF_{\eta\mid\varphi}} \left[ \lambda_2 \E_{q(\eta \mid \varphi)}
M(w \mid \theta) +\mathcal{D}_{\eta\mid\varphi}\{q(\eta\mid\varphi) \| p(\eta\mid\varphi)\} \right].
\end{align*}
Then the optimization-centric posterior density in \eqref{ocpost1} is the OCCP defined in \eqref{cut_varphi}.
\end{lemma}

\begin{proof}
For the given loss, objective function of the optimization in \eqref{ocpost1} is 
\begin{equation*} 
\begin{aligned}
&\E_{q(\varphi)} \{\lambda_1 L(z\mid \varphi)\} + 
	\mathcal{D}_\varphi\{ q(\varphi) \| p(\varphi) \}  \\
& + \E_{q(\varphi)} \left[ \lambda_2 \E_{q(\eta\mid \varphi)} \{ M(w\mid \varphi, \eta)\} + \mathcal{D}_{\eta\mid \varphi} \{ q(\eta \mid  \varphi) \| p(\eta\mid \varphi)\}  + N(w\mid\varphi) \right] .
\end{aligned}
\end{equation*}
For a given $\varphi$, the expression in $[\cdot]$ in the second line is minimized by $q(\eta\mid\varphi)=p^*(\eta\mid \varphi,w)$ $\forall \, \varphi$. The minimum value attained is 0, due to the definition of $N(w\mid\varphi)$. The second line is 0 after minimization with respect to $q(\eta\mid \varphi)$, regardless of $q(\varphi)$. Hence the optimal $q(\varphi)$ minimizes the first line, leading to $q(\varphi)=p^*_\cut(\varphi\mid z)$. 
\end{proof}

\begin{rem}
\upshape The sequential definition of the OCCP in \eqref{cut_varphi} is more intuitive as it shows how the second module's influence is removed when inferring the first module's parameters. This is less obvious in Lemma \ref{lem:cutpost}, but viewing the OCCP as an optimization-centric posterior allows us to use associated theory  and methodology to analyze the OCCP, and provides an elegant unification of both approaches.
\end{rem}

\begin{rem}
\upshape When $\mathcal{D}_{\eta\mid \varphi}$ is KLD and $\mathcal{F}$ is unrestricted, $p^*(\eta \mid \varphi,w)$ is the Gibbs posterior in \eqref{mfn}. Taking log followed by the expectation of both sides of \eqref{mfn} with respect to $q(\eta \mid \varphi)$, we can see that for $q(\eta\mid\varphi)=p^*(\eta\mid\varphi, w)$, 
\[
\lambda_2 \E_{q(\eta \mid \varphi)}\{M(w\mid \theta)\} 
+ \mathcal{D}_{\eta\mid \varphi} \{ q(\eta \mid  \varphi) \| p(\eta\mid \varphi)\}=-\log m_\eta(w\mid\varphi),
\]
which shows that $N(w\mid\varphi)=\log m_\eta(w \mid \varphi)$. Thus, the choice of loss in Lemma \ref{lem:cutpost} is consistent with that in \cite{Frazier2024}, and  \cite{Carmona2020} in the case of cutting feedback for parametric models, where $m_\eta(w \mid \varphi)=p(w \mid \varphi)$.
\end{rem}

From characterization of the OCCP in \eqref{cut_varphi} and Lemma \ref{lem:cutpost}, $p_\cut^*(\theta \mid y)$ can be found via a two-stage procedure. Let $\mL_1 = \lambda_1 \mL_{1, \ls} + \mL_{1, \dv}$ where $\mL_{1, \ls} = \E_{q(\varphi)} \{ L(z \mid \varphi)\}$ and $\mL_{1, \dv} = \mathcal{D}_\varphi\{ q(\varphi) \| p(\varphi) \}$, and $\mL_2 = \lambda_2 \mL_{2, \ls} + \mL_{2, \dv}$ where $\mL_{2, \ls} =  \E_{q(\theta)} \{M(w\mid \theta) \}$ and $\mL_{2, \dv} = \E_{q(\varphi)}[\mathcal{D}_{\eta\mid \varphi} \{ q(\eta \mid  \varphi) \| p(\eta\mid \varphi) \}]$. First, we obtain $p_\cut^*(\varphi \mid z)$ by minimizing $\mL_1$ with respect to $q(\varphi) \in \mcF_\varphi$. Second, we minimize $\mL_2$ with respect to $q(\eta\mid \varphi) \in \mcF_{\eta\mid \varphi}$ by plugging in $q(\varphi) = p_\cut^*(\varphi \mid z)$.

\subsection{R\'{e}nyi divergence and choice of learning rate}
Introducing learning rates in the OCCP provides flexibility in adjusting scale of the expected loss relative to the divergence, but prior robustness cannot be achieved simply by adjusting learning rates. This is inappropriate and may lead to poor uncertainty quantification, as
it is not just magnitude of the divergence that matters, but how the divergence measure varies across the space of candidate densities, and its relative weighting of different densities. \cite{Knoblauch2022} further discuss robust divergences in optimization-centric generalized Bayesian inference.

We consider the use of R\'{e}nyi divergence in the OCCP, due to its robustness and ability to moderate prior influence. R\'{e}nyi divergence is defined as 
\begin{equation*}
\mathcal{D}_\alpha(q \| p) = 
\frac{1}{\alpha-1} \log \left( \int q(\theta)^\alpha p(\theta)^{1-\alpha} \,\dt\theta \right)
\quad \text{for} \quad \alpha \in \mathbb{R}^+ \setminus \{1\}.
\end{equation*}
As R\'{e}nyi divergence converges to KLD when $\alpha\rightarrow 1$, we define $\mathcal{D}_1(q\|p) = \KL(q\|p) = \int q(\theta) \log\{ q(\theta)/p(\theta)\} \,\dt\theta$. Lemma \ref{lemS1} in Appendix \ref{sec: renyi results} provides useful results for evaluating the R\'{e}nyi divergence between common variational densities and priors. In conventional variational inference, \cite{Li2016} show that variational densities transit from mode-seeking to mass-covering  behavior as $\alpha$ decreases to 0. Moreover, R\'{e}nyi divergence is increasing on $\alpha > 0$, with $\mathcal{D}_\alpha (q \| p) \leq \mathcal{D}_\beta(q\| p)$ if $0 < \alpha < \beta$ \citep{Erven2014}. Thus, as $\alpha$ falls to 0 and the divergence diminishes, the OCCP will favor densities that minimize the expected loss primarily. To counteract this phenomenon, we use learning rates to scale the divergence term.

Learning rate of a KLD regularizer can be chosen by matching information \citep{holmes+w17, Lyddon2019}, achieving a target coverage probability \citep{syring+m18}, or treating learning rates as parameters of an unnormalized model \citep{jewson+r22,yonekura+s23}. \cite{wu+m23} summarize several existing techniques, while \cite{Frazier2024} apply them to a two-module generalized cut posterior with two learning rates. However, we are unaware of any studies addressing learning rate selection for other divergences, and we propose a two-step strategy for setting learning rates with R\'{e}nyi divergence regularizers. First, an appropriate learning rate is selected for KLD. Second, we calibrate the divergence penalty across different values of $\alpha$ to match this baseline for a reference density. We do not advocate any choice of learning rate for KLD, as different approaches are suited to different problems. Specifically, let $p_{\cut, \alpha}^*(\theta \mid y) = p_{\cut, \alpha}^* (\varphi \mid z) p_{\alpha}^*(\eta \mid \varphi, w)$ be the OCCP based on R\'{e}nyi's $\alpha$-divergence. Writing $\lambda_1^\alpha$ and $\lambda_2^\alpha$ as learning rates for R\'{e}nyi's $\alpha$-divergence, we set  
\begin{equation*} 
\lambda_1^\alpha = \frac{\mathcal{D}_\alpha \{p^*_{\cut, 1}(\varphi \mid z) \| p(\varphi)\}}{\mathcal{D}_1 \{p^*_{\cut, 1}(\varphi\mid z) \| p(\varphi)\}}\lambda_1^1, 
\;\;
\lambda_2^\alpha = \frac{\E_{p^*_{\cut, 1}(\varphi \mid z)} [\mathcal{D}_\alpha \{p^*_{1}(\eta \mid \varphi, w) \| p(\eta \mid \varphi)\}]}{\E_{p^*_{\cut,1}(\varphi \mid z)} [\mathcal{D}_1 \{p^*_{1}(\eta \mid \varphi, w) \| p(\eta \mid \varphi)\}]}\lambda_2^1.
\end{equation*}
This makes the prior penalty for KLD and R\'{e}nyi's $\alpha$-divergence the same for any $\alpha$, when evaluated at the OCCP for KLD. If the loss terms are identical and the prior regularization term is scaled correctly relative to the loss for the Gibbs posterior, matching penalties will put the $\alpha$-divergence prior regularization on a proper scale.

Our use of R\'{e}nyi divergence in the OCCP differs from its use in variational inference \citep{Li2016}, which considers in our context the optimization: $ \argmin_{q\in\mathcal{F}_{\theta}} \mathcal{D}_\alpha \{ q(\theta) \| p^*_\cut(\theta \mid y) \}$. The solution to this problem will not coincide with that in (2) unless $\alpha\rightarrow 1$. Hence, for $\alpha\ne1$, our use of R\'{e}nyi divergence does not encompass existing R\'{e}nyi divergence based variational inference, and posterior concentration results in Section \ref{sec posterior concentration} are not related to results on R\'{e}nyi divergence based variational posterior concentration \citep{zhang2020convergence, jaiswal2020asymptotic}.

\section{Examples} \label{sec: examples}
This section discusses general examples where the OCCP can produce inference that is more robust to model misspecification than standard Bayesian inference.

\begin{example}[Nonparametric quasi-likelihoods]
\upshape Inference in generalized linear models (GLMs) can be performed in the OCCP framework using quasi-likelihood based loss functions. Consider a GLM based on the first and second order relations between a response variable $Y$ and covariates $X$ for some known link function $g(\cdot)$, an unknown function of the covariates $\varphi(\cdot)$ and a variance function $V_\eta(\cdot)$ known up to $\eta$. Let $\E(Y \mid X=x) = g^{-1}\{ \varphi(x) \}$ and $\Var(Y \mid X=x) = V_\eta [g^{-1}\{ \varphi(x) \}]$. \cite{linero2024bayesian} proposes to conduct Bayesian inference on $\varphi$, $\eta$ using the quasi-likelihood based generalized posterior, $p(\eta,\varphi\mid y,x)\propto p(\varphi) p(\eta\mid\varphi) \exp\left[-(\lambda/2)\sum_{i=1}^{n}D\{y_i,\varphi(x_i),\eta\}\right]$, where $D \{y,\varphi(x),\eta\} = 2\int_{\varphi}^y \; (y-t) / V_\eta \{g^{-1}(t)\} \; \dt t$. As joint inference on $\varphi$, $\eta$ can be sensitive to choice of $V_\eta(\cdot)$ and the priors, sequential inference on $\varphi$, $\eta$ can be considered as an alternative using the loss functions, $L(y,x\mid \varphi) = 2\int_{\varphi}^{y} \;(y-t) / V^\star \{g^{-1}(t)\} \; \dt t$ for some fixed and known $V^\star(\cdot)$ and $M(y,x\mid \eta,\varphi) = D\{y,\varphi(x),\eta\}$. This example implies that the OCCP fits within the sequential Gibbs posterior framework of \cite{winter2023sequential}, which allows one to conduct Bayesian inference in a sequential manner without the need for nested MCMC sampling techniques.
\end{example}

\begin{example}[Copula modeling]
\upshape Copulas are multivariate models that can capture general dependence structures in a random vector $y\in\mathbb{R}^m$, and are specified through two modules: marginal model densities, $f(y_j\mid\varphi_j)$, $j=1,\dots,m$, and a copula density $c\{F_1^{-1}(y_1 \mid \varphi_1),\dots,F_m^{-1}(y_m\mid\varphi_m)\mid\eta\}$, where $F_j^{-1}(y_j\mid\varphi_j)$ denotes the inverse cumulative distribution function (cdf) of the $j$th marginal model. The parameters $\varphi_j$ govern the behavior of the $j$th marginal model, while $\eta$ encodes the dependence structure of the assumed copula density $c(\cdot\mid \eta)$. This model falls in the OCCP framework by setting $z=(y_1,\dots,y_m)^\top$,  $L(z\mid\varphi) = \prod_{j=1}^m f(y_j\mid\varphi_j)$ and $M(w \mid \theta) = c\{F_1^{-1}(y_1 \mid \varphi_1),\dots, F_m^{-1}(y_m \mid \varphi_m) \mid \eta\}$, with $w=(F_1^{-1}(y_{1}),\dots,F_m^{-1}(y_{m}))^\top$.  
\end{example}

\subsection{Biased normal means} \label{sec biased mean}
To further illustrate the advantages of the OCCP, we return to the biased means example \citep{Bayarri2009} in Section \ref{sec-intro}. Let $z_i \mid \varphi \sim \N(\varphi, 1)$ for $i=1, \dots, n_1$ and $w_j  \mid \theta \sim \N(\varphi+\eta, 1)$ for $j=1, \dots, n_2$ independently, with $\bar{z} = n_1^{-1}\sum_{i=1}^{n_1} z_i$ and $\bar{w} = n_2^{-1}\sum_{j=1}^{n_2} w_j$. While $z$ is a reliable data source, $w$ is possibly biased with unknown bias $\eta$. Independent priors, $\varphi \sim \N(\mu_0, v_0)$ and $\eta \sim \N(0, v_b)$ are assigned.

\cite{Bayarri2009} raised two issues with basing inference on $p(\varphi \mid y)$, assuming $v_0 \rightarrow \infty$. First, comparing $\Var(\varphi \mid y) = \{n_1 + 1/(v_b + 1/n_2)\}^{-1}$ and $\Var(\varphi \mid z) = n_1^{-1}$, taking into account the biased data on top of the reliable data does not reduce the marginal posterior variance of $\varphi$ drastically, even if $n_2$ is large, unless $v_b$ is small. Second, if the prior for $\eta$ is misspecified, it can be detrimental to base inference on $\E(\varphi \mid y) = \{n_1 \bar{z} + n_2 \bar{w} / (n_2 v_b+1) \} \{n_1 + n_2 / (n_2 v_b+1) \}^{-1}$ instead of $\E(\varphi \mid z) = \bar{z}$. If the bias is large but $v_b$ is small, too much weight will be allocated to $\bar{w}$, causing the overall posterior mean to deviate far from $\bar{z}$, which is more reliable.

Using cutting feedback to prevent information flow from $w$ to $\varphi$, the marginal posterior of $\varphi$ based only on reliable data $z$ is $\varphi \mid z \sim \N \left( s_1 \bar{z} + (1-s_1) \mu_0, s_1/n_1\right)$, where $s_1 = n_1 v_0/(n_1 v_0 + 1)$ and $s_1\rightarrow 1$ as $v_0\rightarrow\infty$. In addition, $\eta \mid \varphi, w \sim \N\left(s_2(\bar{w} - \varphi), s_2/n_2\right)$, where $s_2 = n_2 v_b/(n_2 v_b + 1)$. To apply the OCCP framework, suppose the loss functions are the negative log-likelihoods and $\lambda_1^1 = \lambda_2^1 = 1$ for KLD. We assume $q(\varphi) = \N(\mu_\varphi, v_\varphi)$ and $q(\eta \mid \varphi) = \N(\mu_\eta + v_{\varphi, \eta} (\varphi - \mu_\varphi) / v_\varphi, v_\eta - v_{\varphi, \eta}^2 / v_\varphi)$, which correspond to a bivariate normal distribution where $\varphi$ and $\eta$ have marginal means $\mu_\varphi$ and $\mu_\eta$, marginal variances $v_\varphi$ and $v_\eta$, and covariance $v_{\varphi,\eta}$. When $\alpha=1$, the OCCP is just the cut joint posterior as R\'{e}nyi divergence reduces to KLD. If $\alpha \neq 1$, optimal values of $\{\mu_\varphi, \mu_\eta, v_\phi, v_\eta, v_{\varphi, \eta}\}$ are found in two stages. First, we find $\{\mu_\varphi, v_\varphi\}$ that minimize $\mL_1 = \lambda_1^\alpha \mL_{1, \ls} + \mL_{1, \dv}$. Second, we find $\{\mu_\eta, v_\eta, v_{\eta, \varphi}\}$ that minimize $\mL_2 = \lambda_2^\alpha\mL_{2, \ls} +  \mL_{2, \dv}$ by fixing $\{\mu_\varphi, v_\varphi\}$ at their optimal values in stage 1.

We simulate 1000 data sets from the model, setting $n_1=20$, $n_2=1000$, $\varphi = \eta =5$ and $v_b=1$. There is more biased than reliable data and the prior variance for $\eta$ is too small. For $\varphi$, we specify an objective prior $\N(0, 100)$ or a strong subjective prior $\N(10, 1)$ that erroneously pulls the posterior mean away from the true value 5. We compute the true posterior, cut posterior and OCCP for $0.05 \leq \alpha \leq 5$ using the BFGS algorithm in {\tt R}. Henceforth, $\alpha=0.999$ represents results under KLD approximately. For each data set, we compute the bias and root mean squared error (RMSE) of the posterior means relative to the true values. We also find the 95\% credible intervals and coverage probabilities of containing the true values. These measures and learning rates averaged over all data sets are given in Table \ref{Tab1}.
\begin{table}[tb!]
\centering
\resizebox{\columnwidth}{!}{\begin{tabular}{lc|rrrrr|rrrrr}
\hline
&& \multicolumn{5}{c|}{Objective prior $\varphi \sim \N(0, 100)$} & \multicolumn{5}{c}{Subjective prior $\varphi \sim \N(10, 1)$} \\
&$\alpha$ & True & 0.05 & 0.5 & 0.999 & 5 
& True & 0.05 & 0.5 & 0.999 & 5 \\ \hline
\multirow{2}{*}{Bias} & $\varphi$ & 0.2354 & -0.0036 & -0.0027 & -0.0025 & \bf {-0.0023} & 0.4543 & {\bf 0.0711} & 0.2095 & 0.2381 & 0.2298 \\ 
&$\eta$ & -0.2392 & {\bf -0.0007} & -0.0015 & -0.0015 & -0.0016 
& -0.4579 & {\bf -0.0754} & -0.2134 & -0.2419 & -0.2334 \\ 
\hline
\multirow{2}{*}{RMSE} & $\varphi$ & 0.3198 & {\bf 0.2272} & 0.2273 & 0.2273 & 0.2273
& 0.4992 & {\bf 0.2382} & 0.3060 & 0.3219 & 0.3159  \\ 
& $\eta$ & 0.3227 & {\bf 0.2273} & 0.2274 & 0.2274 & 0.2274 
& 0.5025 & {\bf 0.2396} & 0.3088 & 0.3248 & 0.3187 \\ 
\hline
\multirow{2}{*}{Coverage} & $\varphi$ & 0.808 & {\bf 0.978} & 0.949 & 0.942 & 0.935 
& 0.429 & {\bf 1.000} & 0.995 & 0.807 & 0.642 \\ 
& $\eta$ & 0.821 & {\bf 0.977} & 0.955 & 0.949 & 0.938 
& 0.422 & {\bf 1.000} & 0.994 & 0.819 & 0.638 \\ 
\hline
\multirow{2}{*}{\begin{minipage}{1.2cm}
Learning \\ rate
\end{minipage}} & $\varphi$ && 0.69 & 0.94 & 1.00 & 1.09 
&& 0.50 & 0.94 & 1.00 & 1.06 \\ 
& $\eta$ && 0.92 & 0.99 & 1.00 & 1.02 
&& 0.91 & 0.99 & 1.00 & 1.02 \\ \hline
\end{tabular}}
\caption{Results of true posterior and OCCPs in simulation. Best values are in bold. \label{Tab1}}
\end{table}

The bias and RMSE of the OCCPs are smaller than that of the true posterior for both priors, showing the benefits of cutting feedback. For the subjective prior, R\'{e}nyi divergence with a small $\alpha$ is very robust and effective in reducing the bias and RMSE. Coverage probabilities of the OCCPs are uniformly higher than that of the true posterior, and increase steadily as $\alpha$ decreases. As $\alpha$ increases, R\'{e}nyi divergence increases and the learning rates increase in tandem to scale the expected loss appropriately. Figure \ref{Fig1} shows the contour plots for a simulated data set under the subjective prior. The OCCPs can capture the true values better than the true posterior and R\'{e}nyi divergence with a smaller $\alpha$ yields better performance. Variance of the OCCP is largest at $\alpha=0.05$, and decreases as $\alpha$ increases, which explains the higher coverage probabilities for smaller $\alpha$ in Table \ref{Tab1}.

\section{Concentration of OCCP} \label{sec posterior concentration}
This section extends theoretical results on cutting feedback methods to OCCPs and allows for more general parameter spaces. We work with OCCPs defined according to the normalized losses, $\Ln_{n_1}(\varphi) =\frac{1}{n_1}L(z\mid\varphi)$ and $\Mn_{n_2}(\theta) = \frac{1}{n_2}M(w\mid\theta)$. Define the functions  $\Ln_{n_1}(\varphi,\varphi')=\Ln_{n_1}(\varphi)-\Ln_{n_1}(\varphi')$ and $\Mn_{n_2}(\theta,\theta')=\Mn_{n_2}(\theta)-\Mn_{n_2}(\theta')$. Let $C$ be a non-random positive constant that can change from line-to-line, and for non-stochastic sequences $a_n$, $b_n$, let $a_n\lesssim b_n$ mean that $a_n \le C b_n$ for all $n$ large enough, and $a_n\asymp b_n$ if $a_n\lesssim b_n$ and $b_n\lesssim a_n$.  Let $\E $ denote expectation with respect to the true data generating process (DGP), $(w_i,z_i) \stackrel {iid}{\sim}P_0$. With generic expectations, given a function $f(X)$, let $\E_{G}[f(X)]$ denote the expectation of $f(X)$ when $X\sim G$. For $\theta,\theta'\in\Theta$, let $\dt(\cdot,\cdot)$ denote some distance operating on the space of $\Theta$.

\subsection{Useful Bounds for OCCPs}
Our first results are formulated as provably approximately correct (PAC) expectation inequalities for the expected loss associated with the OCCP, often referred to as PAC-Bayes inequalities (see e.g., \citealp{Alquier2024}). These results rely on the following two regularity conditions.

\begin{assumption}\label{ass:loss_opt}
(i) The functions $\varphi\mapsto \Ln(\varphi):=\E\Ln_{n_1}(\varphi)$ and $\theta\mapsto \Mn(\theta):=\E\Mn_{n_2}(\theta)$ exist; (ii) For any $\delta>0$, there exist $\varphi^\star\in\Phi$ and $\eta^\star\in\mathcal{E}$ such that $\inf_{\varphi: \; \dt(\varphi,\varphi^\star) > \delta}\Ln(\varphi)>\Ln(\varphi^\star)$ and $\inf_{\eta:\; \dt(\eta,\eta^\star)>\delta}\Mn(\eta,\varphi^\star)>\Mn(\eta^\star,\varphi^\star)$. 
\end{assumption}

Define the functions $\Ln(\varphi,\varphi') = \Ln(\varphi)-\Ln(\varphi')$, and $\Mn(\theta,\theta')=\Mn(\theta)-\Mn(\theta')$. 

\begin{assumption}\label{ass:loss_moments}
For any $\lambda_1,\lambda_2>0$, there exists constants $C_1$ and $C_2$ such that \\
(i) $\int \E \left[e^{{\lambda_1} \{ \Ln(\varphi,\varphi^\star)-\Ln_{n_1} (\varphi,\varphi^\star)\}} \right]p(\varphi)\dt\varphi\le e^{\lambda_1^2C_1^2/{{n_1}}}$; \\
(ii) $\E \int e^{\lambda_2 \{\Mn(\theta,\theta^\star)-\Mn_{n_2}(\theta,\theta^\star)\}} p(\eta\mid\varphi) p_\cut(\varphi\mid z) \, \dt\eta \, \dt\varphi\le e^{\lambda_2^2C_2^2/n_2}$.
\end{assumption}

\begin{rem}
\upshape Assumption \ref{ass:loss_moments}(i) is a standard moment condition in PAC-Bayes analysis \citep[e.g. Thm 2.4,][]{alquier2020concentration}. Whether it holds depends on the choice of loss function and prior, and properties of the true DGP. {Assumption \ref{ass:loss_moments}(ii) requires moments with respect to the modified prior $p(\eta\mid\varphi)p_\cut(\varphi\mid z)$ instead of the prior as in standard PAC-Bayes analysis, because inference in OCCPs is conditional on $\varphi$. Hence, the assumptions must account for the conditional relationship of $\eta\mid\varphi$. } 
\end{rem}

\begin{lemma}\label{lem:PAC-Bayes1}
Under Assumptions \ref{ass:loss_opt}(i) and \ref{ass:loss_moments}(i), the following statements hold:
\begin{flalign}\label{eq:varphi_PAC}
&\E\left[\int \Ln(\varphi,\varphi^\star)\dt p_\cut(\varphi\mid z)\right]
\le\frac{\lambda_1C_1^2}{n_1}\\
&+\E \inf_{q\in\mathcal{F}_\varphi}
\begin{cases}
\int \Ln_{n_1}(\varphi,\varphi^\star)\dt q(\varphi)+\mathcal{D}_\alpha(q\| p)/\lambda_1, &\alpha \geq 1,\;\lambda_1>0;	\\
\int \Ln_{n_1}(\varphi,\varphi^\star)\dt q(\varphi)+\wh\lambda_1{\mathcal{D}_\alpha(q\|p)}{}, &0<\alpha<1,\;\wh\lambda_1\ge\lambda_1^{-1}\frac{\KL(p^*_\cut\|p)}{\mathcal{D}_{\alpha}(p^*_\cut\|p)}.
\end{cases}\nonumber
\end{flalign}	
\end{lemma}

\begin{lemma}\label{lem:PAC-Bayes2}	
Under Assumptions \ref{ass:loss_opt} and \ref{ass:loss_moments}, the following statements hold:
\begin{flalign}
&\E \left[\int \Mn(\theta,\theta^\star)\dt p_\cut(\theta\mid y)\right] \le\frac{\lambda_2C_2^2}{n_2} + \label{eq:eta_PAC} \E \inf_{q\in\mathcal{F}_{\eta\mid\varphi}} \E_{\varphi\sim\wh{q}} \\
&
\begin{cases}
\int \Mn_{n_2}(\theta,\theta^\star)\dt q(\eta\mid\varphi)
+ \mathcal{D}_\alpha\{q(\eta\mid\varphi)\|p(\eta\mid\varphi)\} / \lambda_2
& \alpha \geq 1,\;\lambda_2>0;\\
\int \Mn_{n_2}(\theta,\theta^\star)\dt q(\eta\mid\varphi)+\wh\lambda_2{ \mathcal{D}_\alpha\{q(\eta\mid\varphi)\|p(\eta\mid\varphi)\}}{},& 0<\alpha<1,\;\wh\lambda_2\ge\lambda^{-1}_2\frac{\KL(p_\cut^*\|p)}{\mathcal{D}_\alpha(p_\cut^*\|p)}.
\end{cases}	\nonumber
\end{flalign}	
\end{lemma}

Lemma \ref{lem:PAC-Bayes1} states a PAC-Bayes inequality on the OCCP for $\varphi$. Standard PAC-Bayes results only treat joint optimization and KLD, whilst our OCCP treat sequential optimization and other divergences. Thus, the regularity conditions and result for OCCP deviate from the usual PAC-Bayes conditions. Lemma \ref{lem:PAC-Bayes2} contains a novel variation of the usual PAC-Bayes inequality that leverages sequential nature of the OCCP. Lemmas \ref{lem:PAC-Bayes1} and \ref{lem:PAC-Bayes2} are the first to our knowledge to treat Gibbs posteriors that are specified as a sequential optimization problem and built using divergences besides KLD. As simulations in Section \ref{sec biased mean} have shown, the choice of divergence has an impact on behavior of the OCCP, which is reflected in nature of the result. Lemmas \ref{lem:PAC-Bayes1} and \ref{lem:PAC-Bayes2} also provide a measure of control on behavior of the expected integrated loss associated with the OCCP. As anticipated, the results differ by choice of $\alpha$ for the R\'{e}nyi divergence, and are specific to the sequential nature of the OCCP.

\subsection{Posterior Concentration}
The control given by Lemmas \ref{lem:PAC-Bayes1} and \ref{lem:PAC-Bayes2} allows us to deduce posterior concentration rates for different components in the OCCP, and more importantly to understand how $\lambda_1$ and $\lambda_2$ must be chosen to ensure this posterior concentration (see Remark \ref{rem:inter} for details). To simplify the results, we set $n_1= n_2= n$. Extending such results to $n_1\ne n_2$ requires minor but tedious changes.

To deduce a posterior concentration rate, we require regularity conditions akin to the prior mass condition employed in Bayesian nonparametrics, but specific to our variational inference setting. While we treat general loss functions, divergences other than KLD, and a different variational optimization problem, the conditions employed in Theorem 2.6 of \cite{alquier2020concentration} to deduce posterior concentration of variational posterior approximations for likelihood loss functions, provide a convenient scaffold to generalize such regularity conditions to our setting.

\begin{assumption}\label{ass:prior_varphi}
There exists a sequence $\epsilon_{1,n}$ such that (i) there is a non-random distribution $\rho_{1,n}\in\mathcal{F}_\varphi$ for which $\int \Ln(\varphi,\varphi^\star)\dt \rho_{1,n}(\varphi)\lesssim \epsilon_{1,n}$; (ii) when $\alpha \in (0,1)$, $\exists$ a random variable $W^\alpha_n$ with $\E(W^\alpha_n)\le C$ for all $n$ large enough, such that
$\frac{\KL\{\wh{q}(\varphi)\|p(\varphi)\}}{\mathcal{D}_\alpha\{\wh{q}(\varphi)\|p(\varphi)\}}\leq W^\alpha_n$ (with probability 1); (iii) $\mathcal{D}_\alpha\{\rho_{1,n}(\varphi)\|p(\varphi)\}\lesssim \lambda_1\epsilon_{1,n}$ $\forall$ $\alpha>0$. 
\end{assumption}

\begin{cor}\label{cor:cut_post1}
Under Assumptions \ref{ass:loss_opt}(i), \ref{ass:loss_moments}(i) and \ref{ass:prior_varphi}, for $\lambda_1=\lambda_{1,n}\rightarrow\infty$ such that $\lambda_{1,n}/(n\epsilon_{1,n})\rightarrow0$, and for some $M$ large enough, possibly $M\rightarrow\infty$, 
$$
\E\left( \mathrm{P}_\cut\left[\left\{\varphi\in\Phi: \Ln(\varphi,\varphi^\star)\ge M\epsilon_{1,n}\right\}\mid z\right]\right)\lesssim 1/M.
$$
\end{cor}

To obtain concentration of the joint OCCP, we require a version of local regularity conditions in Assumption \ref{ass:prior_varphi} that cater for sequential nature of the OCCP of $\eta\mid\varphi$.

\begin{assumption}\label{ass:prior_eta}
Let $\rho_{1,n}$ be as in Assumption \ref{ass:prior_varphi}. There exists a sequence $\epsilon_{2,n}$ such that: (i) there is a non-random distribution $\rho_{2,n}\in\mathcal{F}_{\eta\mid\varphi}$ satisfying $\int \Mn(\theta,\theta^\star)\dt \rho_{2,n}(\eta\mid\varphi)\dt\rho_{1,n}(\varphi)\lesssim \epsilon_{2,n}$; (ii) when $0<\alpha<1$, there exists a random variable $V_n^\alpha(\varphi)$, with $\E\{\E_{\varphi\sim\rho_{1,n}} V_n^\alpha(\varphi)\}\le C$ for all $n$ large enough, such that
$\frac{\KL\{\wh{{q}}(\eta\mid\varphi)\|p(\eta\mid\varphi)\}}{\mathcal{D}_\alpha\{\wh{q}(\eta\mid\varphi)\|p(\eta\mid\varphi)\}}\le V_n^\alpha(\varphi)$; (iii) $\int \dt \rho_{1,n}(\varphi)\mathcal{D}_\alpha\{\rho_{2,n}\|p(\eta\mid\varphi)\}\lesssim \lambda_2\max\{\epsilon_{1n},\epsilon_{2n}\}$ for $\alpha>1$; (iv) $\KL(\rho_{1,n}\times\rho_{2,n}\|p)\lesssim \lambda_2\max\{\epsilon_{1n},\epsilon_{2n}\}$. 
\end{assumption}

\begin{rem}
\upshape Assumptions \ref{ass:prior_eta}(i), (iv) are similar to existing conditions on the optimizer employed in the study of concentration for variational posteriors (\citealp{alquier2020concentration}), but are modified to account for sequential nature of the OCCP. In contrast, Assumptions \ref{ass:prior_eta}(ii), (iii) are specific to R\'{e}nyi divergence. As $\mathcal{D}_\alpha(q\|p)\le \KL(q\|p)$ for $\alpha\in(0,1)$, $V_n^\alpha(\varphi)\ge1$, and Assumption \ref{ass:prior_eta}(ii) stipulates that the cut posterior obtained in the second stage using R\'{e}nyi divergence must not cause the KLD to explode. This condition can often be checked analytically, but validation will depend on the chosen variational family and underlying priors. Assumption \ref{ass:prior_eta}(iii) is the R\'{e}nyi divergence equivalent of the prior mass condition maintained for KLD in Assumption \ref{ass:prior_eta}(iv) and has a similar interpretation. We show that Assumptions \ref{ass:prior_varphi}-\ref{ass:prior_eta} are satisfied for the biased means example in Appendix \ref{sec biased means appendix}.
\end{rem}

\begin{cor}\label{cor:cut_post2}
Assumptions \ref{ass:loss_opt}-\ref{ass:prior_eta} are satisfied.  For $j,k\in\{1,2\}$, $\lambda_j=\lambda_{j,n}\rightarrow\infty$ and $\lambda_{j,n}/(\epsilon_{k,n}n)\rightarrow0$, for some $M$ large enough, possibly $M\rightarrow\infty$, 
$$
\E\left( \mathrm{P}_\cut\left[\left\{\theta\in\Phi\times\mathcal{E}: \Mn(\theta,\theta^\star)\ge M\max(\epsilon_{1n},\epsilon_{2n})\right\}\mid z,w\right]\right)\lesssim 1/M.
$$
\end{cor}

\begin{rem}\label{rem:inter}
\upshape We briefly discuss two key insights from Corollaries \ref{cor:cut_post1} and \ref{cor:cut_post2}. First, by decoupling $\varphi$ and $\eta$, we can choose the learning rate for the first module $\lambda_1$ to deliver the best concentration rate for $\varphi$ without worrying that the learning rate or prior for $\eta$ will impact our inferences. For example, if $\varphi$ represents the parametric components in the model, we can choose $\lambda_1$ to achieve the standard parametric rate and make the inference of $\varphi$ as accurate as possible. \cite{Moss2024} consider such a situation, where asymptotic normality results on $\varphi$
are obtained by separating the parametric and nonparametric components ($\eta$
in our notation) of a hidden Markov model using a clever choice of prior, 
and then performing inference on $\eta\mid\varphi$ using the conditional posterior.
Second, and crucially, if $\eta$ concentrates at a slower rate than $\varphi$, Corollary \ref{cor:cut_post1} shows that this slower convergence rate does not negatively impact inferences on $\varphi$. However, if the situation is reversed, and the rate of convergence for $\varphi$ is slower than that of $\eta$, then our results show that the resulting rate of convergence will generally be sub-optimal due to the diameter of the sets under which concentration is computed, which is the larger of the two rate sequences. As a consequence, Corollary \ref{cor:cut_post2} suggests that one may wish to sequester the components of primary interest into a loss by themselves when possible, so that a fast rate is achievable, and then conduct conditional inference on the components of secondary interest, at a possibly slower rate.
\end{rem}

\begin{rem}
\upshape Inspection of the proof of Corollary \ref{cor:cut_post2} shows that  the choices of $\lambda_1$ and $\lambda_2$, and hence their influence on the rates of convergence, can be decoupled in certain settings. This decoupling will occur primarily when $\Mn_n(\theta)$ only depends on $\eta$ but not $\varphi$; such as if the occurrence of $\varphi$ in the loss $\Mn_n(\theta)$ is replaced by a preliminary estimator of $\varphi$, say $\wh{\varphi}=\int \varphi\wh{q}(\varphi)\dt\varphi$; or if the loss $\Mn_n(\eta,\varphi)$ is chosen so that $\eta$ and $\varphi$ are orthogonal, in the sense that $\int \Mn_n(\eta,\varphi)\wh{q}(\varphi)\dt\varphi=0$ for each $\eta$. In each of these examples,  $\int \Mn(\theta)\wh{q}(\eta\mid\varphi)\wh{q}(\varphi)$ will be well-approximated by $\int \Mn(\eta,\varphi^\star)\wh{q}(\eta\mid\varphi)\wh{q}(\varphi)$. While interesting, formalizing such settings and the resulting behavior requires additional assumptions beyond those imposed by Assumptions  \ref{ass:loss_opt}-\ref{ass:prior_eta} and so we propose to study this setting in future research.
\end{rem}

\section{Removing hidden confounding}  \label{sec hidden confounding}
In this section, we consider using experimental data to remove hidden confounding in models fitted to observational data \citep{Kallus2018}. The goal is to study the effects of a binary treatment using two data sources. The first is a large observational study that covers the target population but is susceptible to hidden confounding, as there may be unrecorded covariates that affect treatment allocation and outcome. The second is a small data set from a randomized controlled trial (RCT), which is unconfounded and has covariate overlap with the observational study, but is not representative of target population. We model the observational data using a GP and the RCT data via a parametric correction. This setting is similar to the biased means example, in that standard Bayesian inference can be misleading if the prior for the nonparametric GP or parametric correction parameters is poorly specified, and such issues can be addressed through the OCCP. Our experiments indicate that R\'{e}nyi divergence with a small $\alpha$ yields credible intervals with higher coverage probabilities of containing the true parameters. As $\alpha$ decreases to zero, the variational density changes its behavior from mode-finding to mass-covering. Hence, there is lower tendency of variance underestimation even when strong independence assumptions are made in the space of candidate posteriors for computational efficiency.

Let $X \in \mathbb{R}^p$, $T\in \{0,1\}$ and $Y\in \mathbb{R}$ denote the covariates, treatment and outcome respectively. The observational data $\{X_i^c, T_i^c, Y_i^c \mid i=1, \dots, n_1\}$ and RCT data $\{X_j^u, T_j^u, Y_j^u \mid j=1, \dots, n_2\}$ are assumed to consist of iid draws from a population with hidden confounding (indicated by event $E^c$) and one that is unconfounded (indicated by event $E^u$) respectively, with $n_1 \gg n_2$. We seek to estimate $\tau(x) = \E(Y \mid T=1, X=x, E^u) - \E(Y \mid T=0, X=x, E^u)$, the {\em conditional average treatment effect} (CATE). Using only the RCT data to estimate CATE is limiting in terms of scale of data and domain of covariates, but the observational data can only produce an estimate of $\omega(x) = \omega_1 (x) - \omega_0 (x)$ where $\omega_t(x) = \E(Y \mid T=t, X=x, E^c)$ for $t=0, 1$. The difference $c_f(x) = \tau(x) - \omega(x)$ is the {\em confounding effect}, which is zero if the observational study is unconfounded.

Let $q(X_j^u) = T_j^u/e^u(X_j^u) - (1- T_j^u)/\{ 1 - e^u(X_j^u) \}$ be a signed reweighting function, where $e^u(x) = \P(T=1 \mid X=x, E^u)$ is the propensity score for the RCT data. Then $w_j = q(X_j^u) Y_j^u$ can be shown to be an unbiased estimate of $\tau(X_j^u)$. \cite{Kallus2018} propose to learn $\tau$ by  first obtaining an estimate $\widehat{\omega}(x)$ of $\omega(x)$ from the observational data. Assuming $c_f(x) = \eta^\top x$, $\eta$ is then estimated by $\widehat{\eta} = \argmin_\eta \sum_{j=1}^{n_2} \{ w_j - \widehat{\omega}(X_j^u) - \eta^\top X_j^u \}^2$ and $\tau(x)$ is estimated as $\widehat{\tau}(x) = \widehat{\omega}(x) + \widehat{\eta}^\top x$.

\subsection{An optimization-centric generalized Bayesian extension}
We consider a Bayesian analogue based on cutting feedback. Let $X = (\tilde{X}^\top, \breve{X}^\top)^\top$ be a partitioning of $X$ into continuous variables $\tilde{X} \in \mathbb{R}^{p_1}$ and dummy variables $\breve{X} \in \{0,1\}^{p_2}$ encoding the categorical variables. Let
\[
\omega_t (x) = \delta_t^\top \breve{x} + \zeta_t (\tilde{x})
\quad \text{for} \quad  t = 0, 1,
\]
where the effects of categorical predictors are modeled using linear regression, while zero mean independent GP priors are placed on $a(\tilde{x}) = \{\zeta_1(\tilde{x}) + \zeta_0(\tilde{x}) \} /2$ and $d(\tilde{x}) = \{\zeta_1(\tilde{x}) - \zeta_0(\tilde{x}) \}/2$. We consider stationary squared exponential covariance functions,
$k_a(h) = \sigma_a^2 \exp \left(- \sum\nolimits_{k=1}^{p_1} \ell_{ak}^2 h_{k}^2/2 \right)$ and $k_d(h) = \sigma_d^2 \exp \left(- \sum\nolimits_{k=1}^{p_1} \ell_{dk}^2 h_{k}^2/2 \right)$,
for $a(\cdot)$ and $d(\cdot)$ respectively, where $h = (h_1, \dots, h_{p_1})^\top$ denotes the difference between two inputs, and allow $\sigma_a, \sigma_d$ $\ell_{ak}$, $\ell_{dk}$ to take nonpositive values. The lengthscales $\{\ell_{ak}\}$ and $\{\ell_{dk}\}$ determine the covariance decay rate, and a variable is effectively removed from inference when its lengthscale is zero (automatic relevance determination).

For data from the observational study, we consider the model,
\begin{align*}
Y_i^c &= \omega_{T_i^c} (X_i^c) + \epsilon_i, 
\quad \text{for} \quad i=1, \dots, n_1, 
\end{align*}
where $\epsilon_i \overset{iid}{\sim} \N(0, \sigma^2)$. Let $a^c = (a(\tilde{X}_1^c), \dots, a(\tilde{X}_{n_1}^c))^\top$, $d^c = (d(\tilde{X}_1^c), \dots, d(\tilde{X}_{n_1}^c))^\top$, $\epsilon = (\epsilon_1, \dots, \epsilon_{n_1})^\top$, $z = (Y_1^c, \dots, Y_{n_1}^c)^\top$, $T^c$ be a $n_1 \times n_1$ diagonal matrix with $i$th diagonal element $2T_i^c - 1$, and $\breve{X}^c$ be a $n_1 \times 2p_2$ matrix with $i$th row $((1-T_i^c)\breve{X}_i^{c\top}, T_i^c\breve{X}_i^{c\top})^\top$. Then
$z \mid a^c, d^c, \delta, \sigma^2 \sim \N(\omega^c, \sigma^2 I_{n_1})$, where $\omega^c = a^c + T^c d^c + \breve{X}^c \delta$ and $\delta = (\delta_0^\top, \delta_1^\top)^\top$. We specify the priors, $\delta \sim \N(\mu_\delta^0, \Sigma_\delta^0)$, $\sigma^2 \sim \text{IG}(a^0, b^0)$, $\sigma_a \sim \N(0, v_{\sigma_a}^0)$, $\sigma_d \sim \N(0, v_{\sigma_d}^0)$ and $\ell_{ak} \sim \N(\mu_{\ell_a}^0, v_{\ell_a}^0)$, $\ell_{dk} \sim \N(\mu_{\ell_d}^0, v_{\ell_d}^0)$ for $k=1, \dots, p_1$.

For the RCT data, we consider the model, 
\begin{align*}
w_j &= \omega(X_j^u) + \eta^\top X_j^u + \varepsilon_j, \quad j = 1, \dots, n_2,
\end{align*}
where $\{\varepsilon_j\}$ are zero mean errors, and let the prior on $\eta$ be $\N(\mu_\eta^0, \Sigma_\eta^0)$. Following \cite{Kallus2018}, we find the OCCP by minimizing a sum of squared errors loss, given by $\|w - \omega^u - X^u \eta\|^2$, where $w=(w_1, \dots, w_{n_2})^\top$, $\omega^u = 2d^u + \breve{X}^u(\delta_1- \delta_0)$, $d^u = (d(\tilde{X}_1^u), \dots, d(\tilde{X}_{n_2}^u))^\top$, $\breve{X}^u = (\breve{X}_1^u, \dots, \breve{X}_{n_2}^u)^\top$ and $X^u = (X_1^u, \dots, X_{n_2}^u)^\top$.

\subsection{Sparse GP approximation} \label{sec_sparseGP}
For efficient inference, we consider a sparse GP approximation that relies on inducing inputs, and inducing variables drawn from a standardized GP so that they are independent of the hyperparameters of the covariance functions apriori \citep{Titsias2013}. Let $r_a = (r_{a,1}^\top, \dots, r_{a, M}^\top)^\top$ and $r_d = (r_{d,1}^\top, \dots, r_{d, M}^\top)^\top$  be inducing inputs where $r_{a,i}, r_{d, i} \in \mathbb{R}^{p_1}$, and  $u_a = (s(r_{a,1}), \dots, s(r_{a, M}))^\top$ and $u_d = (s(r_{d,1}), \dots, s(r_{d, M}))^\top$ be inducing variables. The standardized function $s:\mathbb{R}^{p_1} \rightarrow \mathbb{R}$ maps inducing inputs to zero mean GPs with the standardized squared exponential covariance function, $k_s(h) = \exp(-h^\top h/2)$. Thus $p(u_a)=\N(0, K_{u_a})$ where $K_{u_a} = [k_s(r_{a,i}-r_{a,j})]$ is independent of $(\ell_a, \sigma_a)$, and 
$p(a^c \mid u_a, \sigma_a, \ell_a) = \N(K_{a^c, u_a} K_{u_a}^{-1} u_a,\; K_{a^c} - K_{a^c, u_a} K_{u_a}^{-1} K_{u_a, a^c})$, where $K_{a^c, u_a} = \Cov(a^c, u_a)$ and $K_{a^c} = \Cov(a^c)$. Similar results hold for $u_d$ and $d^c$. Let $\varphi = \{\delta, \sigma, a^c, u_a, \sigma_a,  \ell_a, d^c, u_d, \sigma_d,\ell_d\}$ denote all variables in the augmented model for the confounded data. We have $p(z, \varphi) = p(z \mid \varphi) p(\varphi)$, where $p(z \mid \varphi) = p(z \mid a^c, d^c, \delta, \sigma^2)$ and $p(\varphi) =p(\delta) p(\sigma^2) p(a^c \mid u_a, \sigma_a, \ell_a) p(u_a) p(\sigma_a)  p(\ell_a) 
p(d^c \mid u_d, \sigma_d, \ell_d)  p(u_d) p(\sigma_d) p(\ell_d)$.

\subsection{OCCP for sparse GP}
The first stage loss function is $L(z \mid \varphi) = -\log p(z \mid  \varphi)$ and $\lambda_1^1 = 1$ for KLD. We assume $q(\varphi)$ is of the form,
\begin{align*}
q(\delta) q(\sigma^2) 
p(a^c \mid u_a, \sigma_a, \ell_a) q(u_a) q(\sigma_a)
p(d^c \mid u_d, \sigma_d, \ell_d) q(u_d) q(\sigma_d) 
\prod_{k=1}^{p_1} q(\ell_{a,k})q(\ell_{d, k}) ,
\end{align*}
where $q(\delta)$ is $\N(\mu_{\delta}^q, \Sigma_{\delta}^q)$, $q(\sigma^2)$ is $\text{IG}(a^q, b^q)$,  $ q(u_a)$ is $\N(\mu_{u_a}^q, \Sigma_{u_a}^q)$, $q(\sigma_a)$ is $\N(\mu_{\sigma_a}^q, v_{\sigma_a}^q)$, $ q(u_d)$ is $\N(\mu_{u_d}^q, \Sigma_{u_d}^q)$,   $q(\sigma_d)$ is $\N(\mu_{\sigma_d}^q, v_{\sigma_d}^q)$,  $q(\ell_{a,k})$ is $\N(\mu_{\ell_{a,k}}^q, v_{\ell_{a,k}}^q)$ and $q(\ell_{d, k})$ is $\N(\mu_{\ell_{d, k}}^q, v_{\ell_{d, k}}^q)$. The conditional densities of $a^c$ and $d^c$ in $q(\varphi)$ are identical to those in the prior, while other densities are optimal under mean-field assumption for KLD.

The second stage loss function is $M(w \mid \varphi, \eta) = \|w -\omega^u - X^u \eta\|^2/2$. We interpret the learning rate for KLD as a dispersion parameter and set $\lambda_2^1 = \{\sum_{j=1}^{n_2} (w_j - \widehat{\omega}(X_j^u) - \widehat{\eta}^\top X_j^u)^2 / (n_2 - p)\}^{-1}$ based on method of moments. \cite{Agnoletto2025} argue for this estimate based on consistency and good frequentist coverage, over an empirical estimate based on information number matching \citep{Lyddon2019}. We estimate $\widehat{\omega}(X_j^u)$ using its posterior mean based on results from stage 1, and $\widehat{\eta}$ by minimizing the resulting loss. A learning rate of $\lambda_2^1$ for KLD is equivalent to assuming $\epsilon_j \sim \N(0, 1/\lambda_2^1)$ and the true conditional posterior in this case is $\eta \mid w, \varphi \sim \N(\mu_\eta, \Sigma_\eta)$ where $\Sigma_\eta = (\lambda_2^1 X^{u\top} X^u + {\Sigma_\eta^0}^{-1})^{-1}$ and $\mu_\eta = \Sigma_\eta \{ \Sigma_\eta^{0-1} \mu_\eta^0 + \lambda_2^1 X^{u \top}(w - \omega^u)\}$. Hence, we assume $q(\eta \mid \varphi)$ as $\N(\mu_\eta^q, \Sigma_\eta^q)$, where $\mu_\eta^q = a_\eta^q - \lambda_2^1 \Sigma_\eta X^{u \top} \omega^u$ depends on $\varphi$ through $\omega^u$. The variational parameters to be optimized are $a_\eta^q \in \mathbb{R}^p$ and $\Sigma_\eta^q \in \mathbb{R}^{p \times p}$, which are independent of $\varphi$. Variational parameters are optimized using coordinate descent when $\alpha=1$, and gradient descent when $\alpha \neq 1$. For the sparse GP, $M=10$ regularly spaced inducing inputs are used in each case.

\subsection{Simulation study}
Consider the simulation study in \cite{Kallus2018} with an unmeasured confounder $U$. Confounded observations are simulated via $T^c_i \sim \text{Bern}(0.5)$, and $(X_i^c, U_i^c)\mid T_i^c$ follows a bivariate normal distribution with zero mean and unit variance, with $\Cov(X_i^c, U_i^c\mid T_i^c) = T_i^c - 0.5$ for $i=1, \dots, 1000$. Unconfounded observations are generated via $T^u_j \sim \text{Bern}(0.5)$, $X_j^u \sim U[-1,1]$ and $U_j^u \sim \N(0,1)$ for $j=1, \dots, 100$. For both data sets, the outcome is $Y = 1 + T + X + 2TX + 0.5X^2 + 0.75TX^2 + U + 0.5\epsilon$ where $\epsilon \sim \N(0,1)$. The confounding effect $c_f(x) = -x$ and true value of $\eta$ is $-1$. The OCCP for $\alpha \in \{0.01, 0.05, 0.25, 0.999, 2.5\}$ is computed for 50 simulated data sets.

\begin{table}[tb!]
\centering
\resizebox{\columnwidth}{!}{
\begin{tabular}{llrrrrr}
\hline
\multicolumn{2}{c}{$\alpha$} & 0.01 & 0.05 &  0.25 &  0.999 & 2.5 
\\  \hline
\multicolumn{2}{c}{$\lambda_1^\alpha$} & 0.568 $\pm$ 0.007 & 0.742 $\pm$ 0.008 & 0.893 $\pm$ 0.005 & 1.000 $\pm$ 0.000 & 1.057 $\pm$ 0.005 \\
\multicolumn{2}{c}{$\lambda_2^\alpha$} & 0.009 $\pm$ 0.001 & 0.021 $\pm$ 0.003 & 0.035 $\pm$ 0.004 & 0.044 $\pm$ 0.005 & 0.049 $\pm$ 0.005 \\ \hline
\multirow{3}{*}{$\eta$} & Mean & -1.064 & -1.067 & -1.054 & -1.056 & -1.058  \\
& RMSE & 0.862 & 0.863 & 0.864 & 0.863 & 0.862 \\
& Coverage & 0.980 & 0.980 & 0.940 & 0.920 & 0.880 \\
\hline
\end{tabular}}
\caption{Simulation study. Mean and standard deviation of learning rates (first 2 rows). Mean of $\mu_\eta$, RMSE of $\mu_\eta$ from true $\eta$, and coverage probability of true $\eta$ (last 3 rows). \label{tab_GP_lam}}
\end{table}

From Table \ref{tab_GP_lam}, the learning rates decrease as $\alpha$ drops to 0. The averaged posterior mean of $\eta$ is close to the true value $-1$, but the large RMSE indicates significant deviations for individual simulated data sets. This is likely due to the small unconfounded data set and limited domain of $X^u$. The coverage probability based on 95\% credible intervals is higher for $\alpha$ closer to 0, due to the larger estimated posterior variance of $\eta$. Estimating $\lambda_2^1$ using method of moments improved the coverage probability significantly, which was only about 0.5 if $\lambda_2^1 = 1$. For $\omega(x)$ and the CATE $\tau(x)$, Figure \ref{GP_omegtau} shows that RMSE of the posterior predictive means is minimized at $x=0$ and increases gradually as $|x|$ increases. RMSE of $\omega(x)$ is lower than $\tau(x)$ for a wide range of $x$, suggesting that the inability to recover the true $\eta$ accurately likely contributed to the high RMSE of $\tau(x)$. The coverage probability of the true $\tau(x)$ is highest for $\alpha \in \{0.01, 0.05\}$, and decreases steadily as $\alpha$ increases almost everywhere, except close to the border.

\begin{figure}[tb!]
\centering
\includegraphics[width=\textwidth]{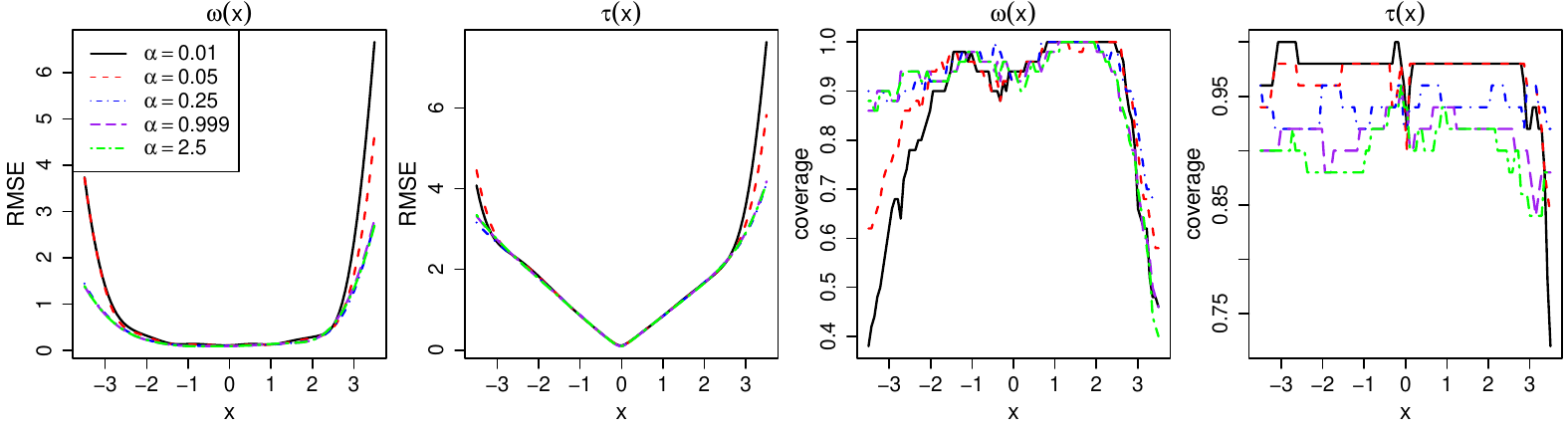}
\caption{\label{GP_omegtau} Simulation study. RMSE of posterior predictive mean from true value, and coverage probability of true value based on 95\% credible intervals for $\omega(x)$ and $\tau(x)$.}
\end{figure}

\subsection{Tennessee STAR (Student/Teacher Achievement Ratio) study}
The Tennessee STAR  study \citep{Krueger1999} was conducted in 1985 to analyze the impact of class size on student outcomes. We consider data from first grade, the sum of the listening, reading and math standardized test scores as outcome $y_i$, and class size as treatment ($T_i=0$ for regular and $T_i=1$ for small). Age (on 1 Jan 1985), gender, race, whether free lunch is given and teacher ID are used as predictors. Omitting observations with missing information, $n=4139$ observations remain. The dimension of $X_i$ is $p=244$, with $p_1 = 1$ and $p_2=243$. As STAR is a RCT, the propensity score $e^u = \text{P}(T=1)$ is independent of predictors, and computed as overall proportion of students assigned to small class size. An unbiased estimate of $\tau(X_i)$ is $q_i Y_i = Y_i/ (e^u + T_i -1)$ for $i=1, \dots, n$. Following \cite{Kallus2018}, we create unconfounded and confounded samples using a variable (rural or urban), that is known to impact student outcome but hidden from analysis. The unconfounded sample is obtained by randomly sampling a fraction $\gamma \in \{0.1, 0.2, \dots, 0.5\}$ of the rural students, whose treatment assignment is random. The confounded sample includes all students with $T=0$ who are not in the unconfounded sample, and students with $T=1$ whose outcomes are in the lower half among rural, and among urban students. Treatment effect estimates in the confounded sample are thus biased downwards. The OCCP is computed for $\alpha \in \{0.01, 0.05, 0.25, 0.999, 2.5\}$.

\begin{figure}[tb!]
\centering
\includegraphics[width=0.95\textwidth]{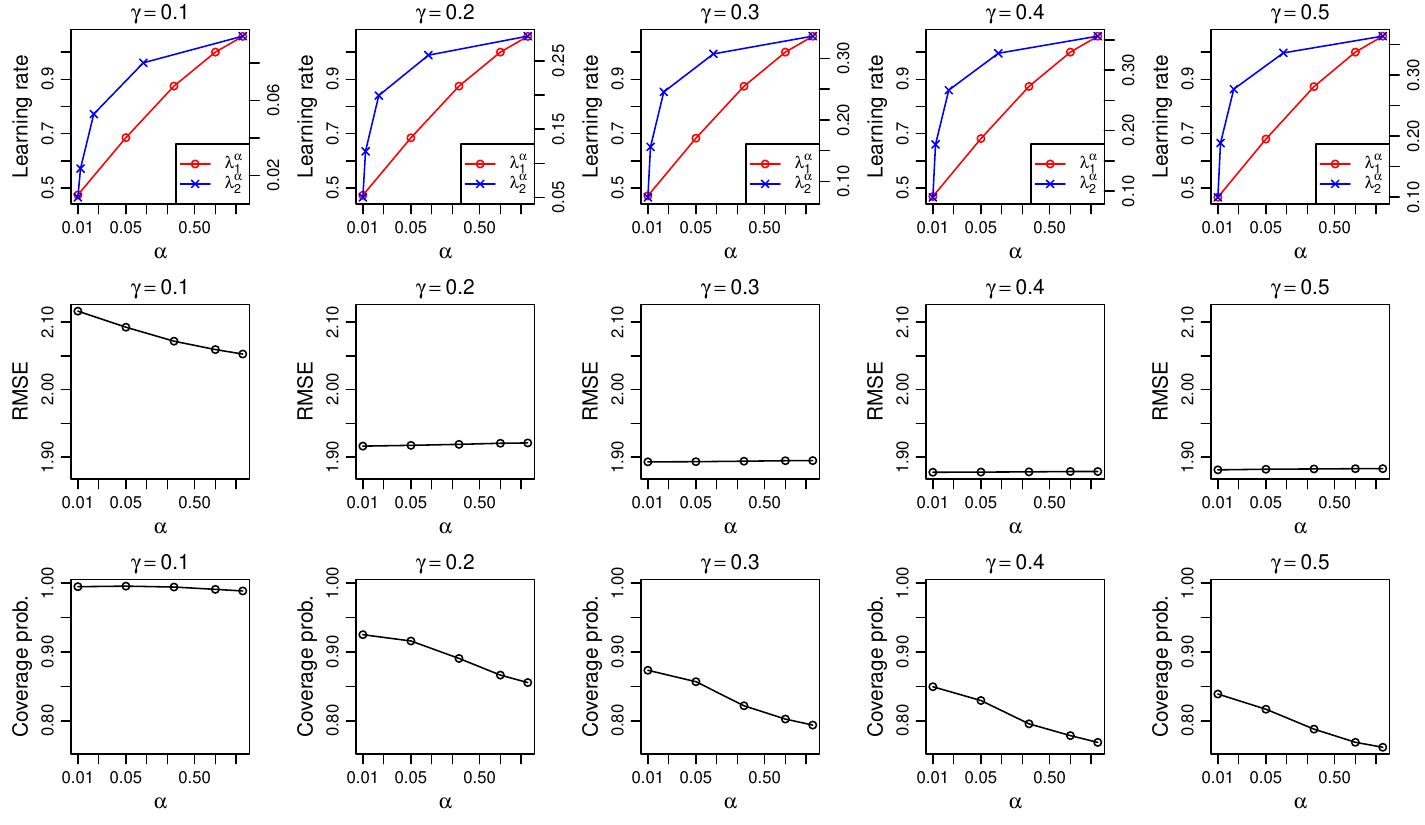}
\caption{\label{fig_star_results} STAR. Learning rates (left axis for $\lambda_1^\alpha$ and right axis for $\lambda_2^\alpha$), RMSE of posterior predictive mean of $\tau(x)$ from its unbiased estimate, and coverage probability of the unbiased estimate of $\tau(x)$ based on 95\% credible intervals.}
\end{figure}

As before, learning rates in Figure \ref{fig_star_results} increase as $\alpha$ increases for each $\gamma$. Learning rates in stage 2 are much lower than in stage 1 and increase with $\gamma$, unlike those in stage 1, which are independent of $\gamma$. This is likely due to the increased precision in estimating $\eta$, as the unconfounded sample size increases with $\gamma$. Comparing the posterior predictive means of $\tau(x)$ with its unbiased estimate on a heldout sample (includes all observations except the unconfounded sample), the RMSE of $\tau(x)$ decreases steadily as $\gamma$ increases, but are almost indistinguishable across $\alpha$, except for a declining trend when $\gamma=0.1$. As $\gamma$ increases, the confounded sample size decreases, and the downward estimation bias is reduced, which leads to a reduction in RMSE. Finally, coverage probability of the unbiased estimate of $\tau(x)$ decreases steadily, as $\alpha$ increases and $\gamma$ increases, as   an $\alpha$ closer to zero tends to yield larger posterior variance estimates. We also observe an improvement in coverage probabilities by estimating $\lambda_2^1$ using method of moments compared to setting it as 1.

\section{Misspecified nonparametric copula model} \label{sec misspecified copula model}
Here, we consider the ``type 1'' cut posterior for copula models \citep{Smith2025}, which prevents misspecification of the copula from contaminating inference on marginal densities. The use of R\'{e}nyi divergence with a small $\alpha$ is shown to produce credible intervals with higher coverage probabilities of containing the true parameters in simulations, and better forecasting in a macroeconomic time series.

Let $\bfy = (\bfy_1^\top, \dots, \bfy_n^\top)^\top$ where $\bfy_i = (y_{i1}, \dots, y_{im})^\top$, $F$ be the joint cdf of $\bfy$, $F_{ij}$ be the marginal cdf of $y_{ij}$ and $N= nm$. By \cite{Sklar1959}, there exists a copula function $C: [0,1]^N \rightarrow [0,1]$ that captures all dependence among $\{y_{ij}\}$, such that $F(\bfy) = C(\bfu)$, where $\bfu = (\bfu_1^\top, \dots, \bfu_n^\top)^\top$, $\bfu_i = (u_{i1}, \dots, u_{im})^\top$ and $u_{ij} = F_{ij}(y_{ij})$. Marginal models for $y_{ij}$ are assumed to be invariant with respect to $i$ so that $F_{ij}(\cdot) = F_j (\cdot)$ $\forall i=1, \dots, n$. We have $p(\bfy \mid \varphi, \eta) = c(\bfu \mid \eta) \prod_{j=1}^m \prod_{i=1}^n  f_{j} (y_{ij} \mid \varphi_j) $, where $c(\bfu \mid \eta) = \frac{\partial}{\partial \bfu} C(\bfu \mid \eta)$ is the copula density, $f_{j} (\cdot \mid \varphi_j)$ is the marginal density of $y_{ij}$, and $\eta$ and $\varphi = (\varphi_1^\top, \dots, \varphi_m^\top)\top$ are associated parameters.

We propose to model the marginal densities $\{f_j\}$ using GPs with a logistic stick-breaking representation and a P\'{o}lya-Gamma (PG) augmentation \citep{Linderman2015}. Suppose $\{y_{ij}\}_{i=1}^n$ are distributed in a finite region $\mathcal{V}_j \subset \mathbb{R}$. This data is first discretized by  partitioning $\mathcal{V}_j$ into $p$ regular bins. The midpoint and number of observations in each bin are stored in $\bfx_j = (x_{j1}, \dots, x_{jp})^\top$ and $\bfn_j = (n_{j1}, \dots, n_{jp})^\top$ respectively. Let $\pi_{jk}$ be the probability that $y_{ij}$ falls in the $k$th bin and $\bfpi_j = (\pi_{j1}, \dots, \pi_{jp})^\top$. Then the likelihood of $\{y_{ij}\}_{i=1}^n$ is $\prod_{k=1}^p \pi_{jk}^{n_{jk}}$. The multinomial density can be expressed as sequential conditional binomial densities such that 
\[
\text{Mult}(\bfn_j \mid n, \bfpi_j) = \prod_{k=1}^{p-1}  \binom{N_{jk}}{n_{jk}}  \tilde{\pi}_{jk}^{n_{jk}} (1 - \tilde{\pi}_{jk})^{N_{jk} - n_{jk}},
\]
where $N_{j1} = n$, $\tilde{\pi}_{j1} = \pi_{j1}$, $N_{jk} = n - \sum_{k' <k} n_{jk'}$ and $\tilde{\pi}_{jk} = \pi_{jk}/(1 - \sum_{k' < k} \pi_{jk'})$ for $k=2, \dots, p-1$. We model the binomial conditional probability as $\tilde{\pi}_{jk} = \exp(g_{jk})/\{1 + \exp(g_{jk})\}$, where $g_{jk}$ is the value of a latent function $g_j(\cdot)$ evaluated at $x_{jk}$ and $\bfg_j = (g_{j1}, \dots, g_{j,p-1})^\top$. For density smoothing, we consider $\bfg_j = \bfd_j + H_j \beta_j$, where $\bfd_j = (d(x_{j1}), \dots, d(x_{j,p-1}))^\top$ and $d_j(\cdot)$ follows a zero mean GP with stationary squared exponential covariance function, $k_j(x, x') = \sigma_j^2 \exp\{-\ell_j^ 2 (x -x')^2/2\}$, while $H_j$ is a $(p-1) \times 3$ matrix where the $k$th row contains the basis functions $(x_{jk}, x_{jk}^2, x_{jk}^3)^\top$ and $\beta_j \in \mathbb{R}^3$ are regression coefficients. We specify the priors, $\sigma_j \sim \N(0, v_{\sigma_j}^0)$, $\ell_j \sim \N(0, v_{\ell_j}^0)$ and $\beta_j \sim \N(0, v_{\beta_j}^0 I_3)$. From the PG augmentation \citep{Polson2013}, 
\begin{align*}
\text{Mult}(\bfn_j \mid n, \bfpi_j) = \prod_{k=1}^{p-1} \binom{N_{jk}}{n_{jk}} 2^{-N_{jk}} \exp(\kappa_{jk}g_{jk}) \int_0^\infty \exp(-\omega_{jk} g_{jk}^2/2) p(\omega_{jk}) d\omega_{jk},
\end{align*}
where $\kappa_{jk} = n_{jk} - N_{jk}/2$, $\omega_j = (\omega_{j1}, \dots, \omega_{j,p-1})^\top$ and each auxiliary variable $\omega_{jk}$ is generated from the prior, $p(\omega_{jk}) =\text{PG}(N_{jk}, 0)$. For efficient inference, we consider the sparse GP approximation in Section \ref{sec_sparseGP} again. Let $r_j = (r_{j,1}, \dots, r_{j, M})^\top \in \mathbb{R}^M$ be inducing inputs and $v_j = (s(r_{j,1}), \dots , s(r_{j,M} ))^\top \in \mathbb{R}^M$ be inducing variables, where $s(\cdot)$ is the standardized function defined in Section \ref{sec_sparseGP} with $p_1 = 1$.

\subsection{OCCP for nonparametric copula model}
The ``type 1'' cut posterior of the copula model falls in our framework, as we can take $p(z \mid \varphi) = \prod_{j=1}^m \prod_{i=1}^n f_j(y_{ij} \mid \varphi_j)$ and $p(w \mid \theta) = c(\bfu \mid \eta)$. With the discretization of $\bfy$ and formation of bins with midpoints $\{\bfx_j\}_{j=1}^m$, we only require the bin counts $z = \{\bfn_j\}_{j=1}^m$ for fitting marginal densities in stage 1. Given the fitted marginal densities, we only need $w$, which comprises the indicators of which bins observations are falling in, to compute $\bfu$ for inferring the copula parameter $\eta$ in stage 2.

For the first stage, let $\varphi_j = (\bfd_j^\top, v_j^\top, \ell_j, \sigma_j, \beta_j^\top, \omega_j^\top)^\top$ for $j=1, \dots, m$. The prior $p(\varphi) = \prod_{j=1}^m p(\varphi_j)$ where $p(\varphi_j) = p(\bfd_j \mid v_j, \ell_j, \sigma_j) p(v_j)  p(\ell_j) p(\sigma_j)  p(\beta_j) \prod_{k=1}^p p(\omega_{jk})$. The loss function is $L(z \mid \varphi) = - \log p(z \mid \varphi) =  -\sum_{j=1}^m (\bfg_j^\top \kappa_j - \frac{1}{2} \bfg_j^\top \Omega_j \bfg_j) + C$, where $\Omega_j = \diag(\omega_j)$, $\kappa_j = (\kappa_{j1}, \dots, \kappa_{j,p-1})^\top$ and $C$ is a constant independent of $\varphi$. $\mL_{1, \ls}$ can be evaluated analytically with the aid of the PG augmentation. We assume $q(\varphi) = \prod_{j=1}^m q(\varphi)$, where $q(\varphi_j) = p(\bfd_j \mid v_j, \ell_j, \sigma_j) q(v_j)  q(\ell_j) q(\sigma_j) q(\beta_j) \prod_{k=1}^{p-1} q(\omega_{jk})$, $q(v_j)$ is $\N(\mu_{v_j}^q, \Sigma_{v_j}^{q \top})$, $q(\ell_j)$ is $\N(\mu_{\ell_j}^q, v_{\ell_j}^{q \top})$, $q(\sigma_j)$ is $\N(\mu_{\sigma_j}^q, v_{\sigma_j}^{q \top})$, $q(\beta_j) $ is $\N(\mu_{\beta_j}^q, \Sigma_{\beta_j}^q)$ and $q(\omega_{jk})$ is $\PG(b_{jk}, c_{jk}^q)$. The assumed variational densities are optimal under mean-field assumption for KLD. R\'{e}nyi divergence can also be evaluated analytically except for $\mathcal{D}_\alpha \{p(\omega_{jk}) \| q(\omega_{jk})\}$. Under KLD, $b_{jk}$'s optimal value can be deduced as $N_{jk}$, but it is difficult to optimize for $\alpha \neq 1$, as the PG density does not have a closed form. Hence we only consider KLD for this pair.

The second stage loss function is $M(w \mid \theta) = - \log c(\bfu \mid \psi)$, where $u_{ij} = \sum_{\ell = 1}^k \pi_{j\ell}$ if $y_{ij}$ falls in the $k$th interval. We consider the prior $\N(0, \Sigma_{\eta}^0)$ for $\eta$ after  transforming copula parameters to the real line. Let $\xi$ denote variables in $\varphi$ excluding $\bfd = (\bfd_1, \dots, \bfd_m)^\top$. We assume $q(\theta) = p(\bfd \mid \xi)q(\xi) q(\eta \mid \xi)$, where $p(\bfd \mid \xi)$ and $q(\xi)$ can be inferred from stage 1, while $q(\eta \mid \xi)$ is the conditional Gaussian $\N(\mu_\eta + \Sigma_{\eta, \xi} \Sigma_\xi^{-1} (\xi-\mu_\xi), \Sigma_\eta - \Sigma_{\eta, \xi} \Sigma_\xi^{-1} \Sigma_{\eta, \xi}^\top)$. Fixing $\{\mu_\xi, \Sigma_\xi\}$ at optimal values obtained in stage 1, only $\{ \mu_\eta, \Sigma_{\eta, \xi}, \Sigma_\eta \}$ are optimized in stage 2.

Variational parameters are optimized using coordinate descent if $\alpha=1$ and gradient descent if $\alpha\neq 1$ in stage 1, while stochastic gradient descent is used in stage 2. Constrained observations are transformed onto the real line before fitting the marginal densities, and the number of bins $p = \min(n/2, 400)$. For the sparse GP, $M=10$ regularly spaced inducing points are used. We set $\lambda_1^1 = \lambda_2^1 = 1$ and the variance for normal priors is 100.

\subsection{Misspecified copula simulation}
In this simulation based on \cite{Smith2025}, we generate 60 data sets, each containing $\{\bfy_i \in \mathbb{R}^2 \mid i=1, \dots, 1000\}$, simulated independently from a $t$-copula with one degree of freedom and Kendall's $\tau$ of 0.7. The first marginal is $\text{Lognormal}(1,1)$ and the second marginal is $\text{Gamma}(7, 3)$. We fit a misspecified Gumbel copula whose parameter is Kendall's tau, $\tau \sim U[0,1]$, transformed as $\eta = \Phi^{-1}(\tau) \sim \N(0,1)$. As \cite{Smith2025} have demonstrated the superiority of the type I cut posterior with correctly specified marginals for this example, we focus on comparing the OCCPs with $\alpha \in \{0.1, 0.25, 0.5, 0.99\}$.

\begin{table}[tb!]
\centering
\small
\begin{tabular}{l|lll|ll}
\hline
& Bias & RMSE & Coverage & \multicolumn{2}{c}{Coverage} \\
$\alpha$ & $\eta$ & $\eta$ & $\eta$ & $f_1$ & $f_2$ \\
\hline
0.1 & 0.0232 & 0.0454 & 0.82 & 0.517 & 0.536 \\
0.25 & 0.0236 & 0.0455 & 0.80 & 0.509 & 0.527 \\
0.5 & 0.0237 & 0.0455 & 0.80 & 0.500 & 0.516 \\
0.999 & 0.0243 & 0.0458 & 0.80 & 0.498 & 0.523 \\
\hline
\end{tabular}
\caption{Simulation. Bias, RMSE and coverage probability of true $\eta$ (first 3 columns) and coverage probability of true marginal densities based on 95\% credible intervals (last 2 columns). \label{tab_copula_sim}}
\end{table}

First we assess the ability of the misspecified copula to capture correlation as measured by Kendall's $\tau$, by computing the bias and RMSE of $\mu_\eta$ from the true value of $\Phi^{-1}(0.7)$, and the coverage probability based on 95\% credible intervals, averaged over 60 simulations. From Table \ref{tab_copula_sim}, the bias and RMSE decrease steadily, while the coverage probability increases as $\alpha$ decreases. The coverage probability of the true marginal density, averaged over all bins and simulated data sets also increases as $\alpha$ decreases, with $\alpha=0.1$ having the highest coverage. Figure \ref{Fig4} shows the discretized data and estimated marginal densities of a simulated data set for $\alpha=0.1$.

\begin{figure}[tb!]
\centering
{\includegraphics[width=\textwidth]{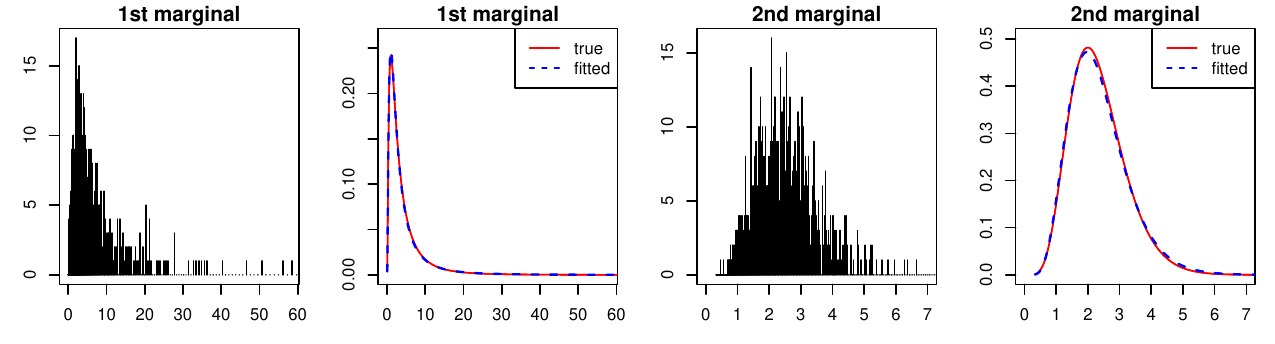}}
\caption{\label{Fig4}  Discretization of marginal observations and the true and fitted marginal densities.}
\end{figure}

\subsection{Gaussian copula for macroeconomic time series}
Next, we study a U.S. macroeconomic time series \citep{Smith2016} with $m=4$ variables (output growth, inflation, unemployment rate and interest rate), observed at $n=274$ quarters, with $N = 1096$ observations. A Gaussian copula is fitted, where $c(\bfu \mid \eta) = |\Omega|^{-1/2} \exp \{ -\bfw^\top (\Omega^{-1} - I_N) \bfw /2\}$, $\bfw = \Phi^{-1}(\bfu)$ and $\Omega \in \mathbb{R}^{N \times N}$ is a correlation matrix. Assuming $\bfw = \{\bfw_s\}_{s=1}^n$ follows a stationary autoregression with order $o=4$, $\Omega$ has a sparse block Toepltiz structure. For efficiency, $\Omega$ is parametrized in terms of its partial correlations $\{ \phi_{ij} \mid i,j=1, \dots, N, i> j \}$ through a vine representation of the copula. The copula parameter $\eta$ consists of 70 unique partial correlations which are transformed onto the real line via $\tilde{\phi}_{ij} = \Phi^{-1} [(\phi_{ij}+1)/2]$, with the prior, $\tilde{\phi}_{ij} \sim \N(0, v_\phi^0)$. As \cite{Smith2025} have demonstrated that the type I cut posterior produces better forecasts than the conventional posterior by using skew-$t$ marginals for this data, we focus on comparing the OCCPs for $\alpha \in \{0.1, 0.25, 0.5, 0.999\}$. Figure 5 shows the discretized data set, fitted marginal densities for $\alpha=0.1$ and 95\% credible intervals computed using 1000 simulations. 
\begin{figure}[tb!]
\centering
{\includegraphics[width=\textwidth]{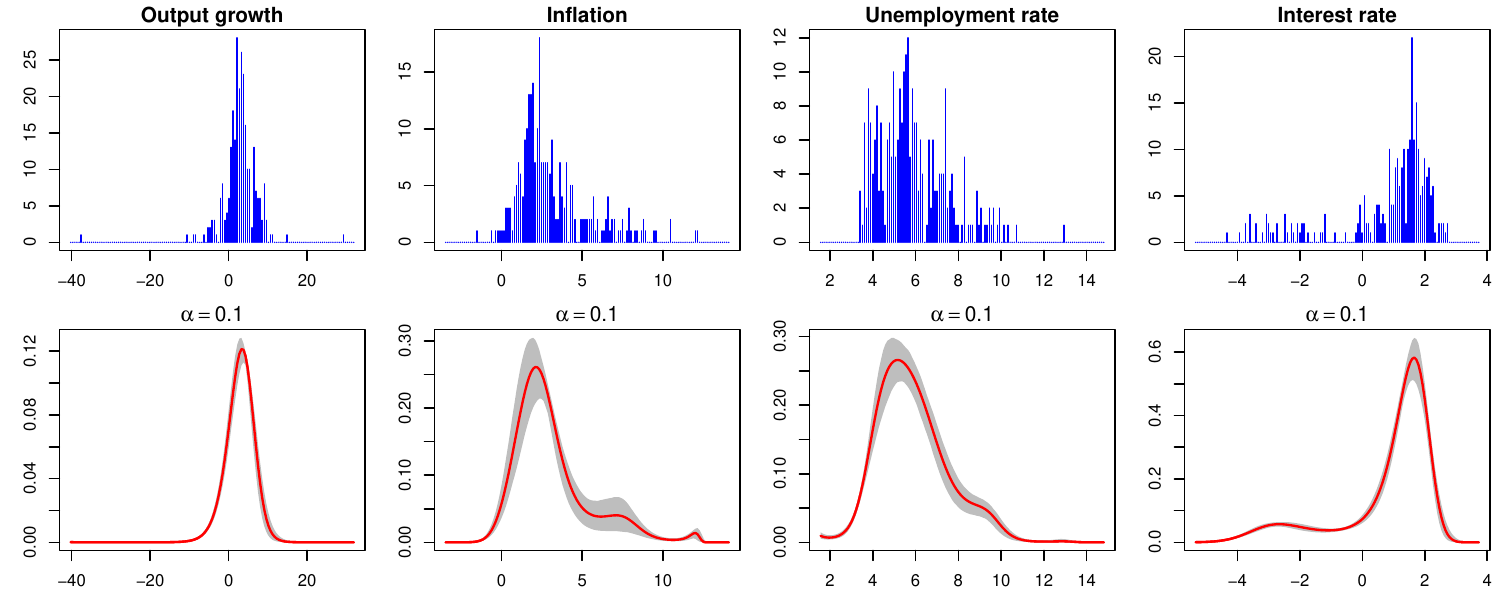}}
\caption{\label{Fig5} Discretization of observations and fitted marginal densities. 95\% credible intervals are shaded in grey.}
\end{figure}

Let $\pi_t(\theta)$ denote the OCCP based on $\{\bfy_s\}_{s=1}^t$. To assess the accuracy of the OCCP, we consider the $h$-steps ahead posterior predictive density,  
\begin{align}\label{pred_density}
f_{t+h \mid t} (\bfy_{t+h} \mid \bfy_t, \dots, \bfy_{t-o+1}) = \int p(\bfy_{t+h} \mid \bfy_t, \dots, \bfy_{t-o+1}, \theta) \pi_t(\theta) d\theta,
\end{align}
for $h=1, \dots, 8$. 5000 draws are generated from \eqref{pred_density} by sampling $\theta \sim \pi_t(\theta)$, and then sampling from $p(\bfy_{t+h} \mid \bfy_t, \dots, \bfy_{t-o+1}, \theta)$ \citep{Smith2016}. A metric $\log \widehat{f} _{t+h \mid t} (y_{j, t+h} \mid \bfy_t, \dots, \bfy_{t-o+1})$ is then computed based on a kernel density estimate of the log of draws from \eqref{pred_density}, evaluated at the observed $y_{j, t+h}$ $\forall j$. Higher values of this metric indicate better posteriors for performing predictions. We compute the OCCP based on the first $t=266$ observations, reserving the last $8$ observations for prediction. This is a challenging task, especially for the first variable (output growth) as there are extreme shocks near the end. From Figure \ref{Fig6}, $\alpha=0.1$ produces the highest score almost uniformly across all variables and time steps, illustrating the robustness of the R\'{e}nyi divergence and advantages of $\alpha<1$ for forecasting.

\begin{figure}[tb!]
\centering
{\includegraphics[width=0.9\textwidth]{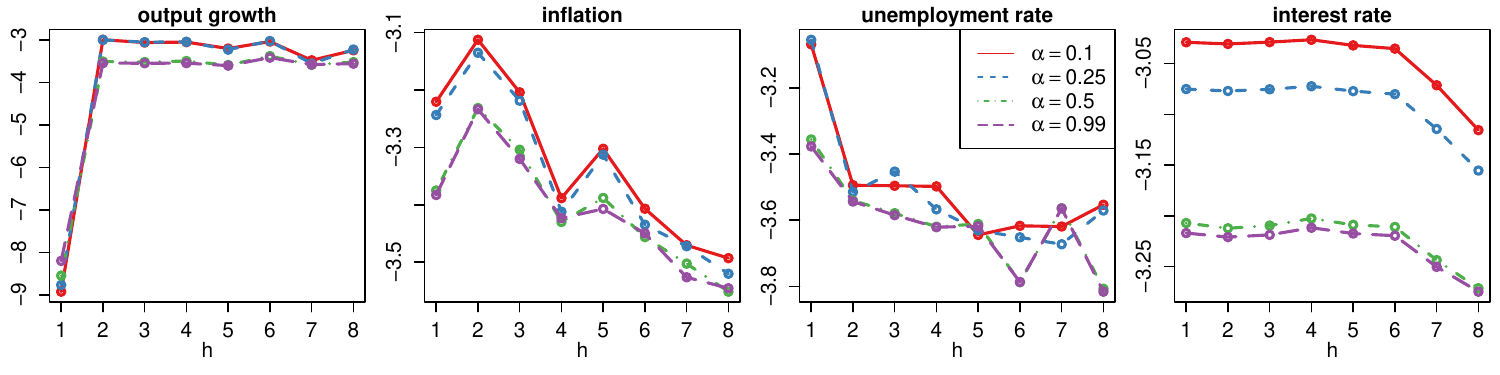}}
\caption{\label{Fig6} Log posterior predictive density estimate for $h=1, \dots, 8$ quarters ahead.}
\end{figure}

\section{Conclusion} \label{sec conclusion}
In this article, we propose an ``optimization-centric'' generalized Bayes approach to cutting feedback for semiparametric joint models, which addresses the risk of likelihood misspecification in parametric modules, and influential prior choices for high-dimensional parameters in nonparametric modules. This highly flexible framework allows possibly misspecified likelihoods to be replaced by general loss functions, more robust divergence measures for capturing prior-data conflicts, and posterior substitutes that can be computed more efficiently. We provide novel theory showing posterior concentration for generalized cut posteriors using PAC-Bayesian bounds, when there are infinite-dimensional parameters in the modules. Our approach is demonstrated using a benchmark example and challenging real applications involving removing hidden confounding in causal inference, and misspecified copula models with nonparametric marginals. Our studies indicate that R\'{e}nyi divergence with a small $\alpha$ is very effective in reducing bias, providing better coverage and improving predictions, relative to KLD. The introduction of learning rates is also essential in appropriate scaling of the prior-posterior penalty relative to the loss function, leading to better coverage.


\acks{Linda Tan's research was supported by the Ministry of Education, Singapore, under the Academic Research Fund Tier 2 (Award MOE-T2EP20222-0002). David Nott's research was supported by the Ministry of Education, Singapore, under the Academic Research Fund Tier 2 (MOE-T2EP20123-0009).}

\appendix

\section{Useful results on R\'{e}nyi divergence}\label{sec: renyi results}
In this appendix, we present a lemma, which contains useful results on evaluating the R\'{e}nyi divergence between commonly encountered variational densities and prior distributions.

\newpage
\begin{lemma} \label{lemS1}
\begin{enumerate}[(a)]
\item[]
\item Suppose $\theta =  (\theta_1^\top, \dots, \theta_n^\top)^\top$, $p(\theta) = \prod_{i=1}^n p_i(\theta_i)$ and $q_i(\theta) = \prod_{i=1}^n q(\theta_i)$. Then $\mathcal{D}_\alpha(q \| p) = \sum_{i=1}^n \mathcal{D}_\alpha(q_i \| p_i)$.

\item Suppose $\theta = (\theta_1^\top, \theta_2^\top)^\top$, $p(\theta) = p(\theta_1\mid \theta_2)p(\theta_2)$ and $q(\theta) = p(\theta_1\mid \theta_2) q(\theta_2)$. Then 
\begin{equation*}
\mathcal{D}_\alpha( q(\theta) \| p(\theta)) = \mathcal{D}_\alpha(q(\theta_2) \| p(\theta_2)).
\end{equation*}

\item Suppose $p(\theta)$ and $q(\theta)$ belong to the same exponential family with natural parameters $\lambda_0$ and $\lambda$ respectively, such that $p(\theta) = h(\theta) \exp\{\lambda_0^\top t(\theta) - A(\lambda_0)\}$ and $q(\theta) = h(\theta) \exp\{\lambda^\top t(\theta) - A(\lambda)\}$. Let  $\lambda' = \alpha \lambda + (1-\alpha) \lambda_0$. Then 
\[
\begin{aligned}
\mathcal{D}_\alpha(q\|p) &= \frac{A(\lambda') - \alpha A(\lambda) -  (1-\alpha) A(\lambda_0) }{\alpha-1}, 
\quad 
\nabla_\lambda \mathcal{D}_\alpha(q\|p) &= \frac{\alpha \{\nabla_\lambda A(\lambda') - \nabla_\lambda A(\lambda)\} }{\alpha-1}.
\end{aligned}
\]

\item Let $\theta$ be a $d$-dimensional random vector, $q(\theta)$ and $p(\theta)$ denote the densities of $\N(\mu, \Sigma)$ and $\N(\mu_0, \Sigma_0)$ respectively, and $W= \alpha \Sigma_0 + (1-\alpha) \Sigma$. Then 
\begin{small}
\begin{equation*}
\begin{aligned}
\mathcal{D}_\alpha(q\|p) = \begin{cases}
\dfrac{\alpha \log|\Sigma_0| }{2(\alpha-1) } - \dfrac{\log|\Sigma| }{2} + \dfrac{\log|W| }{2(1-\alpha) }+  \dfrac{\alpha (\mu - \mu_0)^\top W^{-1} (\mu - \mu_0)}{2} & \text{if $\alpha \in \mathbb{R}^+ \negthickspace\setminus \negthickspace \{1\}$}, \\
\left\{\log|\Sigma^{-1} \Sigma_0| + (\mu-\mu_0)^\top \Sigma_0^{-1} (\mu - \mu_0) + \tr(\Sigma_0^{-1}\Sigma) - d \right\} /2& \text{if $\alpha = 1$}. 
\end{cases}
\end{aligned}
\end{equation*}
\end{small}
In addition, if $CC^\top$ is a Cholesky decomposition of $\Sigma$, then $\forall \alpha > 0$,
\begin{align*}
\nabla_\mu D_\alpha (q \| p) &= \alpha W^{-1} (\mu - \mu_0), \\
\nabla_{\vec(\Sigma)} D_\alpha (q \| p) &= \vec \left\{ W^{-1} - \Sigma^{-1} + \alpha(\alpha -1) W^{-1}(\mu - \mu_0)(\mu - \mu_0)^\top W^{-1}   \right\}/2, \\
\nabla_{\vech(C)} D_\alpha (q \| p) 
&= 2 \vech( \nabla_{\vec(\Sigma)} D_\alpha (q \| p)  C)\\
&= \vech \left\{ W^{-1} C - C^{-\top} + \alpha (\alpha - 1) W^{-1}(\mu - \mu_0)(\mu - \mu_0)^\top W^{-1} C   \right\}.
\end{align*}

\item Let $q(\theta)$ and $p(\theta)$ denote the densities of $\text{IG}(a, b)$ and $\text{IG}(a_0, b_0)$, respectively. Then
\begin{align*}
& \mathcal{D}_\alpha(q \| p) = 
\begin{cases}
\frac{ \log \Gamma(a') - a' \log b' - \alpha\{ \log \Gamma(a) - a \log b\} }{\alpha-1} + \log \Gamma(a_0) - a_0 \log b_0 & \text{if $\alpha \in \mathbb{R}^+ \negthickspace\setminus \negthickspace \{1\}$}, \\
a_0\log \frac{b}{b_0} + \log \frac{\Gamma(a_0)}{\Gamma(a)} + (a - a_0) \psi(a) + a(b_0 - b)/b & \text{if $\alpha = 1$}. 
\end{cases} \\
& \nabla_a D_\alpha(q \| p) = 
\begin{cases}
\frac{ \alpha \{\log b - \log b' +  \psi(a') - \psi(a)\} }{\alpha-1} 
& \text{if $\alpha \in \mathbb{R}^+ \negthickspace\setminus \negthickspace \{1\}$}, \\
(a-a_0) \psi'(a) + b_0/b - 1  & \text{if $\alpha = 1$}.
\end{cases} \\
& \nabla_b D_\alpha(q \| p) = 
\begin{cases}
\frac{ \alpha \{a/b - a' /b'\} }{\alpha-1}
& \text{if $\alpha \in \mathbb{R}^+ \negthickspace\setminus \negthickspace \{1\}$}, \\
a(b - b_0)/b^2  & \text{if $\alpha = 1$}.
\end{cases}
\end{align*}
where $a' = \alpha a + (1-\alpha) a_0 $ and $b' = \alpha b +  (1-\alpha) b_0$.

\item Suppose $p(\theta) = \text{PG}(b, 0)$ and $q(\theta) = \text{PG}(b, c)$ are  P\'{o}lya-Gamma (PG) densities. Then 
\begin{align*}
\mathcal{D}_\alpha(q \| p) &= 
\begin{cases}
\frac{b}{\alpha-1} [\alpha \log\{ \cosh(c/2)\} - \log \{ \cosh(\sqrt{\alpha}c/2) \}] & \text{if $\alpha \in \mathbb{R}^+ \negthickspace\setminus \negthickspace \{1\}$}, \\
b\log \{ \cosh(c/2)\} - \frac{bc}{4} \tanh(c/2)
& \text{if $\alpha = 1$}. 
\end{cases} \\
\nabla_c D_\alpha (q \| p) &= 
\begin{cases}
\frac{b}{2(\alpha-1)}[\alpha \tanh(c/2) -  \sqrt{\alpha} \tanh(\sqrt{\alpha}c/2) ] 
& \text{if $\alpha \in \mathbb{R}^+ \negthickspace\setminus \negthickspace \{1\}$},  \\
\frac{b}{4} \tanh(c/2) -  \frac{bc}{8} \text{sech}^2 (c/2)  
& \text{if $\alpha = 1$}.
\end{cases}
\end{align*}
\end{enumerate}
\end{lemma}

\begin{proof}
\begin{enumerate}[(a)]
\item 
\begin{equation*}
\begin{aligned}
D_\alpha(q \| p) 
&= \frac{1}{\alpha-1} \log \left( \prod_{i=1}^n \int q_i(\theta_i)^\alpha p_i(\theta_i)^{1-\alpha} d\theta_i \right)\\
&= \frac{1}{\alpha-1} \sum_{i=1}^n\log \left(  \int q_i(\theta_i)^\alpha p_i(\theta_i)^{1-\alpha} d\theta_i \right) = \sum_{i=1}^n D_\alpha(q_i \| p_i).
\end{aligned}
\end{equation*}

\item 
\begin{equation*}
\begin{aligned}
D_\alpha(q \| p) 
&= \frac{1}{\alpha-1} \log \int \int \{ p(\theta_1 \mid \theta_2) q(\theta_2) \}^\alpha \{ p(\theta_1 \mid \theta_2) p(\theta_2) \}^{1-\alpha}  d\theta_1 \theta_2 \\
&=  \frac{1}{\alpha-1} \log \int \left( \int p(\theta_1 \mid \theta_2) d\theta_1 \right) q(\theta_2)^\alpha p(\theta_2)^{1-\alpha}  \theta_2 \\
&=  \frac{1}{\alpha-1} \log  \int  q(\theta_2)^\alpha p(\theta_2)^{1-\alpha}  \theta_2 
= D_\alpha(q(\theta_2) \| p(\theta_2)).
\end{aligned}
\end{equation*}

\item Since $\int h(\theta) \exp \{ \lambda'^\top t(\theta) - A(\lambda')\} d\theta = 1$, 
\begin{align*}
D_\alpha(q \| p) 
&= \frac{1}{\alpha-1} \log \left( \int  h(\theta) \exp \{ \lambda'^\top t(\theta) - \alpha A(\lambda) -  (1-\alpha) A(\lambda_0) \} d\theta \right)\\
&= \frac{A(\lambda') - \alpha A(\lambda) -  (1-\alpha) A(\lambda_0) }{\alpha-1}.
\end{align*}

\item From (c), $A(\lambda) = (\mu^\top \Sigma^{-1} \mu + \log |\Sigma|)/2$ and $\lambda = ((\Sigma^{-1} \mu)^\top, -\vec(\Sigma^{-1})^\top/2)^\top$ for a Gaussian density. Let $V =\{  \alpha\Sigma^{-1} + (1-\alpha) \Sigma_0^{-1} \}^{-1} = \Sigma W^{-1} \Sigma_0 = \Sigma_0 W^{-1} \Sigma$ and $m =V\{ \alpha\Sigma^{-1} \mu + (1-\alpha) \Sigma_0^{-1} \mu_0 \}$. Then 
\begin{align*}
& D_\alpha (q \| p) 
= \frac{m^\top V^{-1} m + \log|V| - \alpha(\mu^\top \Sigma^{-1} \mu + \log|\Sigma|) -  (1 - \alpha)( \mu_0^\top \Sigma_0^{-1} \mu_0 + \log |\Sigma_0| ) }{2(\alpha-1) } \\
& =  \frac{\mu_0^\top \Sigma_0^{-1} \mu_0 + \log |\Sigma_0| }{2} + \frac{\alpha (\mu^\top \Sigma^{-1} \mu + \log|\Sigma|) }{2(1-\alpha)}  + \frac{\log|\Sigma| + \log|\Sigma_0| - \log|W| }{2(\alpha-1) }\\
& \quad - \frac{\alpha^2 \mu^\top \Sigma^{-1} V\Sigma^{-1} \mu + 2 \alpha (1-\alpha) \mu^\top \Sigma^{-1} V  \Sigma_0^{-1} \mu_0 + (1-\alpha)^2 \mu_0^\top \Sigma_0^{-1} V  \Sigma_0^{-1} \mu_0  }{2(1-\alpha)} \\
& = \frac{\alpha \log|\Sigma_0| }{2(\alpha-1) } - \frac{\log|\Sigma| }{2} + \frac{\log|W| }{2(1-\alpha) }+  \frac{(1-\alpha)\mu_0^\top \Sigma_0^{-1} V(V^{-1} \Sigma_0 - (1-\alpha)I) \Sigma_0^{-1} \mu_0 }{2(1-\alpha)}  \\
& \quad + \frac{\alpha \mu^\top \Sigma^{-1} V (V^{-1} \Sigma - \alpha I) \Sigma^{-1} \mu - 2 \alpha (1-\alpha) \mu^\top \Sigma^{-1} V  \Sigma_0^{-1} \mu_0}{2(1-\alpha) } \\
& = \frac{\alpha \log|\Sigma_0| }{2(\alpha-1) } - \frac{\log|\Sigma| }{2} + \frac{\log|W| }{2(1-\alpha) }+  \frac{ \alpha (\mu_0^\top \Sigma_0^{-1} V \Sigma^{-1} \mu_0 + \mu^\top \Sigma^{-1} V \Sigma_0^{-1} \mu - 2 \mu \Sigma^{-1} V  \Sigma_0^{-1} \mu_0)}{2}  \\
& = \frac{\alpha\log|\Sigma_0| }{2(\alpha-1) } - \frac{\log|\Sigma| }{2} + \frac{\log|W| }{2(1-\alpha) }+  \frac{ \alpha (\mu_0^\top W^{-1} \mu_0 + \mu^\top W^{-1} \mu - 2 \mu^\top W^{-1} \mu_0)}{2}  \\
& = \frac{\alpha\log|\Sigma_0| }{2(\alpha-1) } - \frac{\log|\Sigma| }{2} + \frac{\log|W| }{2(1-\alpha) }+  \frac{\alpha (\mu - \mu_0)^\top W^{-1} (\mu - \mu_0)}{2}.
\end{align*}
\begin{align*}
\KL(q \| p) &= \E_q \{ \log q(\theta) - \log p(\theta) \} \\
&= \frac{1}{2} \E_q \{ - \log|\Sigma| - (\theta - \mu)^\top \Sigma^{-1} (\theta - \mu) + \log |\Sigma_0| +  (\theta - \mu_0)^\top \Sigma_0^{-1} (\theta - \mu_0) \} \\
&= \frac{1}{2} \{  \log |\Sigma_0| - \log|\Sigma| - d +  (\mu - \mu_0)^\top \Sigma_0^{-1} (\mu - \mu_0) + \tr(\Sigma_0^{-1} \Sigma)   \}.
\end{align*}
The derivatives can be obtained by applying vector differential calculus \cite{Magnus1999}.

\item From (c), $A(\lambda) = \log \Gamma(a) - a \log b$ where $\lambda = (a, b)^\top$ for $\text{IG}(a, b)$. Thus 
\begin{align*}
D_\alpha(q \| p) &= \frac{\log \Gamma(a') - a' \log b' - \alpha(\log \Gamma(a) - a \log b) - (1-\alpha) \log \Gamma(a_0) - a_0 \log b_0  }{\alpha-1}, 
\end{align*}
\begin{align*}
\KL(q \| p) &= \E_q\{ a\log b - \log \Gamma(a) - (a+1) \log \theta - b/\theta\} \\
& \quad  - \E_q \{a_0\log b_0 - \log \Gamma(a_0) - (a_0+1) \log \theta - b_0/\theta \} \\
&= a\log b- a_0\log b_0 + \log \frac{\Gamma(a_0)}{\Gamma(a)} + (a_0 - a) \{ \log b - \psi(a) \} + \frac{a(b_0 - b)}{b}.
\end{align*}

\item The density function of a $PG(b, c)$ random variable is 
\[
\text{PG}(\theta \mid b, c) = \cosh^b(c/2) \exp(-c^2 \theta /2) \text{PG}(\theta \mid b, 0).
\]
If $\alpha \neq 1$, then 
\begin{align*}
D_\alpha(q \| p) &= \frac{1}{\alpha-1} \log \int \text{PG}(\theta \mid b, 0) \left(\frac{\text{PG}(\theta \mid b, c)}{\text{PG}(\theta \mid b, 0)} \right)^\alpha d\theta \\
&= \frac{1}{\alpha-1} \log \left\{  \cosh^{b \alpha}(c/2)\int \text{PG}(\theta \mid b, 0) \exp(-c^2\alpha \theta/2) d\theta \right\}\\
&= \frac{1}{\alpha-1} \log \left\{  \cosh^{b \alpha}(c/2) \cosh^{-b} (\sqrt{\alpha}c/2) \right\}\\
&= \frac{b}{\alpha-1}[\alpha \log\{ \cosh(c/2)\} - \log \{ \cosh(\sqrt{\alpha}c/2) \}]. 
\end{align*}
If $\alpha=1$, 
\begin{align*}
D_\alpha(q \| p) &= \E_q \{  \log \text{PG}(\theta \mid b, c) - \log \text{PG}(\theta \mid b, 0) \} \\
 &= \E_q  [ b\log \{ \cosh(c/2)\} - c^2 \theta/2 ] \\
 &= b\log \{ \cosh(c/2)\} - \frac{bc}{4} \tanh(c/2).
\end{align*}
\end{enumerate}
\end{proof}

\section{Proofs and assumptions verification}\label{sec proofs and verif}
\subsection{Proofs of main results}
This section contains the proofs of all results stated in the paper.

\begin{proof}[Lemma \ref{lem:PAC-Bayes1}]
	The start of the proof follows the usual PAC-Bayes recipe. From Assumption \ref{ass:loss_moments}(i), we have 
	\begin{flalign*}
		1\ge \E\left\{\E_{\varphi\sim p(\varphi)}\left[e^{\lambda_1\left\{\Ln(\varphi,\varphi^\star)-\Ln_{n_1}(\varphi,\varphi^\star)\right\}-{\lambda_1^2C_1^2/{n_1}}}\right]\right\}.
	\end{flalign*}Applying the Donsker and Varadhan Lemma to the RHS of the above yields
	\begin{flalign*}
		1\ge \E e^{\sup_{q\in\mathcal{M}(\Phi)}\left\{{\lambda_1\left\{\Ln(\varphi,\varphi^\star)-\Ln_{n_1}(\varphi,\varphi^\star)\right\}-{\lambda_1^2C_1^2/n_1}}-\KL(q\|p)\right\}}.	
	\end{flalign*}Applying Jensen's inequality, taking logarithms, and dividing both sides by $\lambda_1$, we have
	\begin{flalign*}
		0\ge \E\sup_{q\in\mathcal{M}(\Phi)} \left\{{\left[\Ln(\varphi,\varphi^\star)-\Ln_{n_1}(\varphi,\varphi^\star)\right]-{\lambda_1^2C_1^2/{n_1}}}-\frac{\KL(q\|p)}{\lambda_1}\right\}.	
	\end{flalign*}Since the above is true for all ${q\in\mathcal{M}(\Phi)}$, it is true for $\wh{q}(\varphi):=p_\cut^\star(\varphi\mid z)$. Hence, we can re-arrange the above as
	\begin{flalign*}
		\E\int \Ln(\varphi,\varphi^\star)\dt\wh{q}(\varphi)	\le \E\left\{\int \Ln_{n_1}(\varphi,\varphi^\star)\dt\wh{q}(\varphi)+\frac{\lambda_1C_1^2}{n_1}+\frac{\KL(\wh{q}\|p)}{\lambda_1}\right\}.	
	\end{flalign*}	
	We now break the result up into cases depending on the value of $\alpha$.\\
	
	\noindent\textbf{Case 1: $\alpha\ge1$.} Since the R\'{e}nyi-divergence is increasing on $\alpha>0$, for $\alpha>1$ we have that, for all $\lambda_1>0$, 
	\begin{flalign*}
		\E\int \Ln(\varphi,\varphi^\star)\dt\wh{q}(\varphi)	\le& \E\left\{\int \Ln_{n_1}(\varphi,\varphi^\star)\dt\wh{q}(\varphi)+\frac{\lambda_1C_1^2}{n_1}+\frac{\KL(\wh{q}\|p)}{\lambda_1}\right\}	
		\\\le& 	\E\left\{\int \Ln_{n_1}(\varphi,\varphi^\star)\dt\wh{q}(\varphi)+\frac{\lambda_1C_1^2}{n_1}+\frac{\mathcal{D}_{\alpha}(\wh{q}\|p)}{\lambda_1}\right\}\\=&\E\inf_{q\in\mathcal{F}_\varphi}\left\{\int \Ln_{n_1}(\varphi,\varphi^\star)\dt {q}(\varphi)+\frac{\lambda_1C_1^2}{n_1}+\frac{\mathcal{D}_{\alpha}({q}\|p)}{\lambda_1}\right\},
	\end{flalign*}	where the last equation follows from the definition of $\wh{q}(\varphi)=p_\cut(\varphi\mid z)$ in equation \eqref{cut_varphi}. If $\alpha\rightarrow1$, then $\mathcal{D}_\alpha\rightarrow\KL$ and the result follows as described. 
	
	\noindent\textbf{Case 2: $0<\alpha<1$.} Consider that
	\begin{flalign*}
		\E\int \Ln(\varphi,\varphi^\star)\dt\wh{q}(\varphi)	\le& \E\left\{\int \Ln_{n_1}(\varphi,\varphi^\star)\dt\wh{q}(\varphi)+\frac{\lambda_1C_1^2}{n_1}+\frac{\KL(\wh{q}\|p)}{\lambda_1}\right\}\\=&
		\E\left\{\int \Ln_{n_1}(\varphi,\varphi^\star)\dt\wh{q}(\varphi)+\frac{\lambda_1C_1^2}{n_1}+\frac{\mathcal{D}_{\alpha}(\wh{q}\|p)}{\lambda_1}\frac{\KL(\wh{q}\|p)}{\mathcal{D}_{\alpha}(\wh{q}\|p)}\right\}.
	\end{flalign*}Since the R\'{e}nyi-divergence is increasing on $\alpha>0$, for $\alpha<1$ we have that, $\KL(q\|p)/\mathcal{D}_\alpha(q\|p)\le1$ for all $q,p\in\mathcal{F}_\varphi$. Taking $\wh\lambda_1\ge\KL(\wh q\|p)/\mathcal{D}_\alpha(\wh q\|p)/\lambda_1$, conclude that 
	\begin{flalign*}
		\E\int \Ln(\varphi,\varphi^\star)\dt\wh{q}(\varphi)	\le&
		\E\left\{\int \Ln_{n_1}(\varphi,\varphi^\star)\dt\wh{q}(\varphi)+\frac{\lambda_1C_1^2}{n_1}+\lambda_1\mathcal{D}_{\alpha}(\wh{q}\|p)\frac{\KL(\wh{q}\|p)}{\mathcal{D}_{\alpha}(\wh{q}\|p)}\right\}\\=&	\E\left\{\int \Ln_{n_1}(\varphi,\varphi^\star)\dt\wh{q}(\varphi)+\frac{\lambda_1C_1^2}{n_1}+\wh\lambda_1\mathcal{D}_{\alpha}(\wh{q}\|p)\right\}\\=&	\E\inf_{q\in\mathcal{F}_\varphi}\left\{\int \Ln_{n_1}(\varphi,\varphi^\star)\dt{q}(\varphi)+\wh\lambda_1\mathcal{D}_{\alpha}({q}\|p)\right\}+\frac{\lambda_1C_1^2}{n_1}.
	\end{flalign*}
	%
\end{proof}

\begin{proof}[Lemma \ref{lem:PAC-Bayes2}]
	By Assumption \ref{ass:loss_moments}(ii), we have 
	\begin{flalign}\label{eq:new_eq1}
		1\ge \E\left\{ \E_{\varphi\sim \wh{q}(\varphi)}\E_{\eta\mid\varphi\sim p(\eta\mid\varphi)}\left[e^{\lambda_2\left\{\Mn(\theta,\theta^\star)-\Mn_{n_2}(\theta,\theta^\star)\right\}-{\lambda_2^2C_2^2/{n_2}}}\right]\right\}.
	\end{flalign}Applying  Donsker-Varadhan with $h_{n_2}(\theta,\theta^\star;\lambda_2)={\lambda_2\left\{\Mn(\theta,\theta^\star)-\Mn_{n_2}(\theta,\theta^\star)\right\}-{\lambda_2^2C_2^2/{n_2}}}$, we have 
	$$
	\log \int e^{h_{n_2}(\theta,\theta^\star;\lambda_2)}\dt p(\eta\mid\varphi)\wh{q}(\varphi)=\sup_{q\in\mathcal{M}(\Phi\times\mathcal{E})}\left\{\int e^{h_{n_2}(\theta,\theta^\star;\lambda_2)}\dt q(\eta,\varphi)-\KL\{q(\eta,\varphi)\|p(\eta\mid\varphi)\wh{q}(\varphi)\}\right\}.
	$$
	Applying the above to the RHS of \eqref{eq:new_eq1} yields
	\begin{flalign*}
		1\ge \E e^{\sup_{q\in\mathcal{M}(\Phi\times\mathcal{E})}\left\{{\lambda_2\left\{\Mn(\theta,\theta^\star)-\Mn_{n_2}(\theta,\theta^\star)\right\}-{\lambda_2^2C_2^2/{n_2}}}-\KL\{q(\eta,\varphi)\|p(\eta\mid\varphi)\wh{q}(\varphi)\}\right\}}.	
	\end{flalign*}Applying Jensen's inequality, taking logarithms, and dividing both sides by $\lambda_2$, we have
	\begin{flalign*}
		0\ge \E{\sup_{q\in\mathcal{M}(\Phi\times\mathcal{E})}\left\{{\left[\Mn(\theta,\theta^\star)-\Mn_{n_2}(\theta,\theta^\star)\right]-{\lambda_2C_2^2/{n_2}}}-\frac{1}{\lambda_2}\KL\{q(\eta,\varphi)\|p(\eta\mid\varphi)\wh{q}(\varphi)\}\right\}}.		
	\end{flalign*}Since the above is true for all $q\in\mathcal{M}(\Phi\times\mathcal{E})$, it is true for 
	$$
	\wh{q}(\theta)=\wh{q}(\eta\mid \varphi)\wh{q}(\varphi):=p^\star_\cut(\eta\mid w,\varphi)p^\star_\cut(\varphi\mid z).$$
	Hence, we can re-arrange the above to obtain
	\begin{flalign}
		&\E\int \Mn(\theta,\theta^\star)\dt \wh{q}(\theta)	\nonumber\\\le &\E\left\{\int \Mn_{n_2}(\theta,\theta^\star)\dt \wh{q}(\theta)+\frac{\lambda_2C_2^2}{n_2}+\frac{1}{\lambda_2}\KL\{\wh{q}(\eta\mid\varphi)\wh{q}(\varphi)\|p(\eta\mid\varphi)\wh{q}(\varphi)\}\right\}\nonumber\\\le &\E\left\{\int \Mn_{n_2}(\theta,\theta^\star)\dt \wh{q}(\theta)+\frac{\lambda_2C_2^2}{n_2}+\frac{1}{\lambda_2}\int \KL\{\wh{q}(\eta\mid\varphi)\|p(\eta\mid\varphi)\}\dt \wh{q}(\varphi)\right\}.		\label{eq:new_eq2}
	\end{flalign}		
	We again break up the proof into the cases $\alpha\ge1$ and $0<\alpha<1$.\\ 
	
	\noindent\textbf{Case 1: $\alpha\ge1$.} Consider first the case where  $\alpha\rightarrow1$, and that, by the definition of $\wh{q}(\eta\mid\varphi)=p^\star_\cut(\eta\mid w,\varphi)$, the RHS of \eqref{eq:new_eq2} is minimized by setting $q(\eta\mid\varphi)=\wh{q}(\eta\mid\varphi)$ at any value of $\varphi$, see Lemma \ref{lem:cutpost}  for a proof of this claim. Hence, the RHS of \eqref{eq:new_eq2} is equivalently stated as
	\begin{flalign*}
		&\E\int \Mn(\theta,\theta^\star)\dt \wh{q}(\theta)	\nonumber \\\le &\E\left\{\int \Mn_{n_2}(\theta,\theta^\star)\dt \wh{q}(\eta\mid\varphi)\wh{q}(\varphi)+\frac{1}{\lambda_2}\int \KL\{\wh{q}(\eta\mid\varphi)\|p(\eta\mid\varphi)\}\dt \wh{q}(\varphi)\right\}+\frac{\lambda_2C_2^2}{n_2}\\=&\E \inf_{q\in\mathcal{F}_{\eta\mid\varphi}}\left[\E_{\wh{q}(\varphi)}\E_{q(\eta\mid\varphi)}\left\{ \Mn_{n_2}(\theta,\theta^\star)\right\}+\frac{1}{\lambda_2}\E_{\wh{q}(\varphi)}\KL\{{q}(\eta\mid\varphi)\|p(\eta\mid\varphi)\}\right]+\frac{\lambda_2C_2^2}{n_2},
	\end{flalign*}	which yields the stated result for KL divergence. 
	
	When $\alpha>1$, use the fact that $\mathcal{D}_\alpha(q\|p)\ge \KL(q\|p)$ for each $q,p\in\mathcal{F}_{\eta\mid\varphi}$ to obtain
	\begin{flalign*}
		&\E\int \Mn(\theta,\theta^\star)\dt \wh{q}(\theta)	\nonumber \\\le &\E\left\{\int \Mn_{n_2}(\theta,\theta^\star)\dt \wh{q}(\eta\mid\varphi)\wh{q}(\varphi)+\frac{1}{\lambda_2}\int \KL\{\wh{q}(\eta\mid\varphi)\|p(\eta\mid\varphi)\}\dt \wh{q}(\varphi)\right\}+\frac{\lambda_2C_2^2}{n_2}\\\le &\E\left\{\int \Mn_{n_2}(\theta,\theta^\star)\dt \wh{q}(\eta\mid\varphi)\wh{q}(\varphi)+\frac{1}{\lambda_2}\int \mathcal{D}_\alpha\{\wh{q}(\eta\mid\varphi)\|p(\eta\mid\varphi)\}\dt \wh{q}(\varphi)\right\}+\frac{\lambda_2C_2^2}{n_2} \\=&\E \inf_{q\in\mathcal{F}_{\eta\mid\varphi}}\left[\E_{\wh{q}(\varphi)}\E_{q(\eta\mid\varphi)}\left\{ \Mn_{n_2}(\theta,\theta^\star)\right\}+\frac{1}{\lambda_2}\E_{\wh{q}(\varphi)}\mathcal{D}_\alpha\{{q}(\eta\mid\varphi)\|p(\eta\mid\varphi)\}\right]+\frac{\lambda_2C_2^2}{n_2},
	\end{flalign*}	where the last line uses the definition of $\wh{q}(\eta\mid\varphi)$  in \eqref{cut_varphi} .
	
	\noindent\textbf{Case 2: $0<\alpha<1$.}
	Similar to the argument in Lemma \ref{lem:PAC-Bayes1}, since the R\'{e}nyi-divergence is increasing on $\alpha>0$, for $\alpha<1$, we have $\KL(q\|p)/\mathcal{D}_\alpha(q\|p)\le1$ for all $q,p\in\mathcal{F}_{\eta\mid\varphi}$. Hence, for $\wh\lambda_2\ge(1/\lambda_2)\KL(\wh q\|p)/\mathcal{D}_\alpha(\wh q\|p)$, we can conclude that 
	\begin{flalign*}
		&\E\int \Mn(\theta,\theta^\star)\dt \wh{q}(\theta)	\nonumber \\\le &\E\left\{\int \Mn_{n_2}(\theta,\theta^\star)\dt \wh{q}(\eta\mid\varphi)\wh{q}(\varphi)+\frac{1}{\lambda_2}\int \KL\{\wh{q}(\eta\mid\varphi)\|p(\eta\mid\varphi)\}\dt \wh{q}(\varphi)\right\}+\frac{\lambda_2C_2^2}{n_2}\\\le &\E\left\{\int \Mn_{n_2}(\theta,\theta^\star)\dt \wh{q}(\eta\mid\varphi)\wh{q}(\varphi)+\wh\lambda_2\int \mathcal{D}_\alpha\{\wh{q}(\eta\mid\varphi)\|p(\eta\mid\varphi)\}\dt \wh{q}(\varphi)\right\}+\frac{\lambda_2C_2^2}{n_2} \\=&\E \inf_{q\in\mathcal{F}_{\eta\mid\varphi}}\left[\E_{\wh{q}(\varphi)}\E_{q(\eta\mid\varphi)}\left\{ \Mn_{n_2}(\theta,\theta^\star)\right\}+\wh\lambda_2\E_{\wh{q}(\varphi)}\mathcal{D}_\alpha\{{q}(\eta\mid\varphi)\|p(\eta\mid\varphi)\}\right]+\frac{\lambda_2C_2^2}{n_2},
	\end{flalign*} where the last line again uses the definition of $\wh{q}(\eta\mid\varphi)$  in \eqref{cut_varphi}.
	
\end{proof}

\begin{proof}[Corollary \ref{cor:cut_post1}]
	Setting $\wh{q}(\varphi):=p^\star_\cut(\varphi\mid z)$, and applying Markov's inequality 
	$$
	\E \wh{Q}\left[\left\{\varphi\in\Phi: \Ln(\varphi,\varphi^\star)\ge M\epsilon_{1,n}\right\}\right]\le \E  \frac{1}{M\epsilon_{1,n}}\int\Ln(\varphi,\varphi^\star)\dt\wh{q}(\varphi). 
	$$
	Applying Lemma \ref{lem:PAC-Bayes1} to the RHS of the above then yields
	\begin{flalign}
		&\E \wh{Q}\left[\left\{\varphi\in\Phi: \Ln(\varphi,\varphi^\star)\ge M\epsilon_{1,n}\right\}\right]\nonumber\\&\le	\frac{1}{M\epsilon_{1,n}}\mathbb{E} \inf _{q \in \mathcal{F}_{\varphi}}\left\{\int \Ln_n\left(\varphi, \varphi^{\star}\right) \mathrm{d} q(\varphi)+\frac{\lambda_1C_1^2}{n}+\frac{ \mathcal{D}_\alpha(q \| p)}{\lambda_1}\right\}.\label{eq:post_conc_varphi}
	\end{flalign}
	
	\noindent\textbf{Case 1: $\alpha\ge1$.}
	Since \eqref{eq:post_conc_varphi}  is true for the infimum, we have, for non-random $\rho_{1,n}$ as in Assumption \ref{ass:prior_varphi}, 
	\begin{flalign*}
		&\E \wh{Q}\left[\left\{\varphi\in\Phi: \Ln(\varphi,\varphi^\star)\ge M\epsilon_{1,n}\right\}\right]\\&\le	\frac{1}{M\epsilon_{1,n}}\mathbb{E} \inf _{q \in \mathcal{F}_{\varphi}}\left\{\int \Ln_n\left(\varphi, \varphi^{\star}\right) \mathrm{d} q(\varphi)+\frac{\lambda_1C_1^2}{n}+\frac{\mathcal{D}_\alpha(q \| p)}{\lambda_1 }\right\}\\&\le\frac{1}{M\epsilon_{1,n}} \E \left\{\int \Ln_n\left(\varphi, \varphi^{\star}\right) \mathrm{d} \rho_{1,n}(\varphi)+\frac{\lambda_1C_1^2}{n}+\frac{\mathcal{D}_\alpha(\rho_{1,n} \| p)}{\lambda_1 }\right\}\\&\le\frac{1}{M\epsilon_{1,n}}\left\{ \epsilon_{1,n}+\frac{\lambda_1C_1^2}{n}+\lambda_1\epsilon_{1,n}/\lambda_1\right\}\\&\lesssim 1/M.
	\end{flalign*}

	\noindent\textbf{Case 2: $0<\alpha<1$.} In this case, \eqref{eq:post_conc_varphi} must be restated as 
	\begin{flalign*}
		&\E \wh{Q}\left[\left\{\varphi\in\Phi: \Ln(\varphi,\varphi^\star)\ge M\epsilon_{1,n}\right\}\right]\nonumber\\&\le	\frac{1}{M\epsilon_{1,n}}\mathbb{E} \inf _{q \in \mathcal{F}_{\varphi}}\left\{\int \Ln_n\left(\varphi, \varphi^{\star}\right) \mathrm{d} q(\varphi)+\frac{\lambda_1C_1^2}{n}+\wh{\lambda}_1 \mathcal{D}_\alpha(q \| p)\right\},
	\end{flalign*}where $\wh\lambda_1=\lambda_1^{-1}\KL\{\wh{q}\|p\}/\mathcal{D}_\alpha\{\wh{q}\|p\}$. Again choose $q(\varphi)=\rho_{1,n}(\varphi)$ we see that 
	\begin{flalign*}
		&\E \wh{Q}\left[\left\{\varphi\in\Phi: \Ln(\varphi,\varphi^\star)\ge M\epsilon_{1,n}\right\}\right]\\&\le	\frac{1}{M\epsilon_{1,n}}\mathbb{E} \inf _{q \in \mathcal{F}_{\varphi}}\left\{\int \Ln_n\left(\varphi, \varphi^{\star}\right) \mathrm{d} q(\varphi)+\frac{\lambda_1C_1^2}{n}+\wh\lambda_1 \mathcal{D}_\alpha(q \| p)\right\}\\&\le\frac{1}{M\epsilon_{1,n}} \E \left\{\int \Ln_n\left(\varphi, \varphi^{\star}\right) \mathrm{d} \rho_{1,n}(\varphi)+\frac{\lambda_1C_1^2}{n}+\wh\lambda_1 \mathcal{D}_\alpha(\rho_{1,n} \| p)\right\}\\&\le\frac{1}{M\epsilon_{1,n}}\left\{ \epsilon_{1,n}+\frac{\lambda_1C_1^2}{n}+\E \wh\lambda_1\epsilon_{1,n}/\lambda_1\right\}\\&\lesssim 1/M,
	\end{flalign*}where the last term follows by applying the definition of $\wh\lambda_1$ and Assumption \ref{ass:prior_varphi}(ii).

\end{proof}

\begin{proof}[Corollary \ref{cor:cut_post2}]
	Let $\epsilon_n\ge0$ be such that $\epsilon_n\rightarrow0$ as $n\rightarrow\infty$. From Markov's inequality, 
	$$
	\E \wh{Q}\left[\left\{\theta\in\Theta: \Mn(\theta,\theta^\star)\ge M\epsilon_{n}\right\}\right]\le \E  \frac{1}{M\epsilon_{n}}\int\Mn(\theta,\theta^\star)\dt\wh{q}(\theta). 
	$$ Since $\Ln(\varphi,\varphi^\star)\ge0$, we have 
	\begin{flalign*}
		\E \wh{Q}\left[\left\{\theta\in\Theta: \Mn(\theta,\theta^\star)\ge M\epsilon_{n}\right\}\right]\le& \frac{1}{M\epsilon_{n}}\E  \int\Mn(\theta,\theta^\star)\dt\wh{q}(\theta)\\\le & \frac{1}{M\epsilon_{n}}\E  \left[\int\left\{\Ln(\varphi,\varphi^\star)+\Mn(\theta,\theta^\star)\right\}\dt\wh{q}(\varphi)\dt\wh{q}(\eta\mid\varphi)\right].
	\end{flalign*}
	Applying the arguments in Lemmas \ref{lem:PAC-Bayes1}-\ref{lem:PAC-Bayes2} to the two terms on the RHS of the above then yields
	\begin{flalign*}
		&\E \wh{Q}\left[\left\{\theta\in\Theta: \Mn(\theta,\theta^\star)\ge M\epsilon_{}\right\}\right]\\\le & \frac{1}{M\epsilon_{n}}\E  \left[\int\left\{\Ln(\varphi,\varphi^\star)+\Mn(\theta,\theta^\star)\right\}\dt\wh{q}(\varphi)\dt\wh{q}(\eta\mid\varphi)\right]\\\le& \frac{1}{M\epsilon_n}\E \bigg{[}\int\left\{\Ln_n(\varphi,\varphi^\star)+\Mn_n(\theta,\theta^\star)\right\}\dt\wh{q}(\varphi)\dt\wh{q}(\eta\mid\varphi)\\&+\frac{\mathcal{D}_\alpha\{\wh{q}(\varphi)\|p(\varphi)\}}{\lambda_1}+\lambda_2^{-1}\int \mathcal{D}_{\alpha}\{\wh{q}(\eta\mid\varphi)\|p(\eta\mid\varphi)\}\dt \wh{q}(\varphi)\bigg{]}+\frac{\lambda_1C_1^2}{M\epsilon_nn}+\frac{\lambda_2C_2^2}{M\epsilon_nn}.
	\end{flalign*}However, from the definition of $p^\star_\cut(\theta\mid w,z)=\wh{q}(\varphi)\wh{q}(\eta\mid\varphi)$ in \eqref{cut_varphi} , 
	\begin{flalign*}
		&\E \wh{Q}\left[\left\{\theta\in\Theta: \Mn(\theta,\theta^\star)\ge M\epsilon_{}\right\}\right]\\\le& \frac{1}{M\epsilon_n}\E \bigg{[}\int\left\{\Ln_n(\varphi,\varphi^\star)+\Mn_n(\theta,\theta^\star)\right\}\dt\wh{q}(\varphi)\dt\wh{q}(\eta\mid\varphi)\\&+\lambda_1\mathcal{D}_\alpha\{\wh{q}(\varphi)\|p(\varphi)\}+\lambda_2\int \mathcal{D}_{\alpha}\{\wh{q}(\eta\mid\varphi)\|p(\eta\mid\varphi)\}\dt \wh{q}(\varphi)\bigg{]}+\frac{\lambda_1C_1^2}{n}+\frac{\lambda_2C_2^2}{n}\\&\le \frac{1}{M\epsilon_n}\E\bigg{\{} \inf_{q_1\in\mathcal{F}_\varphi, q_2\in\mathcal{F}_{\eta\mid\varphi}}\bigg{[}\int\left\{\Ln_n(\varphi,\varphi^\star)+\Mn_n(\theta,\theta^\star)\right\}\dt q_1(\varphi)\dt q_2(\eta\mid\varphi)\\&+\frac{\mathcal{D}_\alpha\{q_1(\varphi)\|p(\varphi)\}}{\lambda_1}+\lambda_2^{-1}\int \mathcal{D}_{\alpha}\{q_2(\eta\mid\varphi)\|p(\eta\mid\varphi)\}\dt q_1(\varphi)\bigg{]}\bigg{\}}+\frac{\lambda_1C_1^2}{M\epsilon_nn}+\frac{\lambda_2C_2^2}{M\epsilon_nn}.
	\end{flalign*}Taking $q_1(\varphi)=\rho_{1,n}(\varphi)$ and $q_2(\eta\mid\varphi)=\rho_{2,n}(\eta\mid\varphi)$ as in Assumptions \ref{ass:prior_varphi}-\ref{ass:prior_eta}, we have that 
	\begin{flalign*}
		&\E \wh{Q}\left[\left\{\theta\in\Theta: \Mn(\theta,\theta^\star)\ge M\epsilon_{n}\right\}\right]\nonumber\\\le& \frac{1}{M\epsilon_n}\E \inf_{q_1\in\mathcal{F}_\varphi, q_2\in\mathcal{F}_{\eta\mid\varphi}}\bigg{[}\int\left\{\Ln_n(\varphi,\varphi^\star)+\Mn_n(\theta,\theta^\star)\right\}\dt q_1(\varphi)\dt q_2(\eta\mid\varphi)\nonumber\\&+\lambda_1\mathcal{D}_\alpha\{q_1(\varphi)\|p(\varphi)\}+\lambda_2\int \mathcal{D}_{\alpha}\{q_2(\eta\mid\varphi)\|p(\eta\mid\varphi)\}\dt q_1(\varphi)\bigg{]}+\frac{\lambda_1 C_1^2}{M\epsilon_nn}+\frac{\lambda_2 C_2^2}{M\epsilon_nn}\nonumber\\\le& \frac{1}{M\epsilon_n}\left\{\int \Ln(\varphi,\varphi^\star)\dt\rho_{1,n}(\varphi)+\int \Mn(\theta,\theta^\star)\dt\rho_{1,n}(\varphi)\dt\rho_{2,n}(\eta\mid\varphi) \right\}+\frac{\lambda_1C_1^2}{M\epsilon_nn}+\frac{\lambda_2C_2^2}{M\epsilon_nn}\nonumber\\&+\frac{\mathcal{D}_\alpha\{\rho_{1,n}(\varphi)\|p(\varphi)\}}{M\epsilon_n\lambda_1}+\frac{\int\mathcal{D}_\alpha\{\rho_{2,n}(\eta\mid\varphi)\|p(\eta\mid\varphi)\}\dt\rho_{1,n}(\varphi)}{M\epsilon_n\lambda_2}.
	\end{flalign*}Then, applying Assumption \ref{ass:prior_varphi}-\ref{ass:prior_eta} yields
	\begin{flalign}
		&\E \wh{Q}\left[\left\{\theta\in\Theta: \Mn(\theta,\theta^\star)\ge M\epsilon_{}\right\}\right]\nonumber\\\le& \frac{1}{M\epsilon_n}\left\{\int \Ln(\varphi,\varphi^\star)\dt\rho_{1,n}(\varphi)+\int \Mn(\theta,\theta^\star)\dt\rho_{1,n}(\varphi)\dt\rho_{2,n}(\eta\mid\varphi) \right\}+\frac{\lambda_1 C_1^2}{M\epsilon_nn}+\frac{\lambda_2 C_2^2}{M\epsilon_nn}\nonumber\\&+\frac{\mathcal{D}_\alpha\{\rho_{1,n}(\varphi)\|p(\varphi)\}}{M\epsilon_n\lambda_1}+\frac{\int\mathcal{D}_\alpha\{\rho_{2,n}(\eta\mid\varphi)\|p(\eta\mid\varphi)\}\dt\rho_{1,n}(\varphi)}{M\epsilon_n\lambda_2}\nonumber\\\le&\frac{1}{M\epsilon_n}\left\{\epsilon_{1,n}+\epsilon_{2,n}\right\}+\frac{\lambda_1 C_1^2}{M\epsilon_nn}+\frac{\lambda_2 C_2^2}{M\epsilon_nn}+\frac{\epsilon_{1,n}}{M\epsilon_n}+\frac{\int\mathcal{D}_\alpha\{\rho_{2,n}(\eta\mid\varphi)\|p(\eta\mid\varphi)\}\dt\rho_{1,n}(\varphi)}{M\epsilon_n\lambda_2}\label{eq:new_joint_eq}.
	\end{flalign}
	The last term in \eqref{eq:new_joint_eq} must be handled on a case-by-case basis depending on the value of $\alpha$. We break this up into two cases: $\alpha\ge1$ and $\alpha\in(0,1)$. \\
	
	\noindent\textbf{Case 1: $\alpha\ge1$.} From the properties of KL divergence, 
	$$
	\int \KL\{\rho_{2,n}(\eta\mid\varphi)\|p(\eta\mid\varphi)\}\dt \rho_{1,n}(\varphi)\le\KL\{\rho_{1,n}\times\rho_{2,n}\|\rho_{1,n}\times p(\eta\mid\varphi)\} ,
	$$so that under Assumption \ref{ass:prior_eta}(iii) 
	$$
	\KL\{\rho_{1,n}\times\rho_{2,n}\|\rho_{1,n}\times p(\eta\mid\varphi)\} \le \lambda_2\epsilon_n.
	$$Hence, the last term in \eqref{eq:new_joint_eq} satisfies  
	\begin{flalign}
	\frac{1}{M\epsilon_n\lambda_2}\int \KL(\rho_{2,n}(\eta\mid\varphi)\|p(\eta\mid\varphi))\dt \rho_{1,n}(\varphi)&\le\frac{1}{M\epsilon_n\lambda_2}\KL\{\rho_{1,n}\times\rho_{2,n}\|\rho_{1,n}\times p(\eta\mid\varphi)\} \nonumber\\&\le \frac{\max\{\epsilon_{1,n},\epsilon_{2,n}\}}{M\epsilon_n}
    \le  \frac{1}{M}.\label{eq:new1}
	\end{flalign}
	Alternatively, if we have that $\alpha>1$, we apply Assumption \ref{ass:prior_varphi}(iv) to deduce
	\begin{flalign}
		\label{eq:new2}
		\frac{\lambda_2}{M\epsilon_n}\int \mathcal{D}_\alpha(\rho_{2,n}(\eta\mid\varphi)\|p(\eta\mid\varphi))\dt \rho_{1,n}(\varphi)\le\frac{\max\{\epsilon_{1,n},\epsilon_{2,n}\}}{M\epsilon_n}.
	\end{flalign}
	
	\noindent\textbf{Case 2: $0<\alpha<1$.} For this case, we note that \eqref{eq:new_joint_eq} must be modified as 
	\begin{flalign*}
		&\E \wh{Q}\left[\left\{\theta\in\Theta: \Mn(\theta,\theta^\star)\ge M\epsilon_{}\right\}\right]\nonumber\\\le& \frac{1}{M\epsilon_n}\left\{\int \Ln(\varphi,\varphi^\star)\dt\rho_{1,n}(\varphi)+\int \Mn(\theta,\theta^\star)\dt\rho_{1,n}(\varphi)\dt\rho_{2,n}(\eta\mid\varphi) \right\}+\frac{\lambda_1 C_1^2}{M\epsilon_nn}+\frac{\lambda_2 C_2^2}{M\epsilon_nn}\nonumber\\&+\E \frac{\wh\lambda_1}{M\epsilon_n}\mathcal{D}_\alpha\{\rho_{1,n}(\varphi)\|p(\varphi)\}+\E \frac{\wh\lambda_2}{M\epsilon_n}\int\mathcal{D}_\alpha\{\rho_{2,n}(\eta\mid\varphi)\|p(\eta\mid\varphi)\}\dt\rho_{1,n}(\varphi)\nonumber\\\le&\frac{(\epsilon_{1,n}+\epsilon_{2,n})}{M\epsilon_n}+\frac{\lambda_1 C_1^2}{M\epsilon_nn}+\frac{\lambda_2 C_2^2}{M\epsilon_nn}\\&+\E \frac{\wh\lambda_1}{M\epsilon_n}\mathcal{D}_\alpha\{\rho_{1,n}(\varphi)\|p(\varphi)\}+\E \frac{\wh\lambda_2}{M\epsilon_n}\int\mathcal{D}_\alpha\{\rho_{2,n}(\eta\mid\varphi)\|p(\eta\mid\varphi)\}\dt\rho_{1,n}(\varphi).
	\end{flalign*}From Assumption \ref{ass:prior_varphi}(ii), and Assumption \ref{ass:prior_eta}(ii), respectively, 
	\begin{flalign}
		\label{eq:new3}
		\E \frac{\wh\lambda_1}{M\epsilon_n}\mathcal{D}_\alpha\{\rho_{1,n}(\varphi)\|p(\varphi)\}+\E \frac{\wh\lambda_2}{M\epsilon_n}\int\mathcal{D}_\alpha\{\rho_{2,n}(\eta\mid\varphi)\|p(\eta\mid\varphi)\}\dt\rho_{1,n}(\varphi)\lesssim \frac{\epsilon_{1,n}+\epsilon_{2,n}}{M\epsilon_n}. 
	\end{flalign}

	\noindent Applying equations \eqref{eq:new1}-\eqref{eq:new3} into equation \eqref{eq:new_joint_eq} we arrive at the bound:
	\begin{flalign*}
		\E \wh{Q}\left[\left\{\theta\in\Theta: \Mn(\theta,\theta^\star)\ge M\epsilon_{}\right\}\right]\lesssim&\frac{(\epsilon_{1,n}+\epsilon_{2,n})}{M\epsilon_n}+\frac{\lambda_1}{M\epsilon_nn}+\frac{\lambda_2}{M\epsilon_nn}.
	\end{flalign*}Taking $\epsilon_n=\max\{\epsilon_{1,n},\epsilon_{2,n}\}$, we have that $(\epsilon_{1,n}+\epsilon_{2,n})/\epsilon_n\lesssim C$, while the hypothesis of the result require $\lambda_1/n\epsilon_{n}\rightarrow0$ and $\lambda_2/n\epsilon_{n}\rightarrow0$, so that for $M$ large enough we have that $\lambda_j/n\epsilon_{n}\lesssim 1/M$ for $j=1,2$, and the result follows.
\end{proof}

\subsection{Verification of Assumptions \ref{ass:prior_varphi} and \ref{ass:prior_eta}: Biased Means}\label{sec biased means appendix}
In this section, we verify Assumptions \ref{ass:prior_varphi} and  \ref{ass:prior_eta} for the biased means example considered in the main paper.
For the biased means example, assume that 
\begin{flalign*}
    \E[z_i]=\varphi^\star,\quad \E[w_i\mid\varphi^\star]=(\eta^\star+\varphi^\star),\quad \E[(z_i-\E[z_i])^2]=1=\E[(w_i-\E[w_i])^2].
\end{flalign*}
Then, we have the following definitions for the component losses in the OCCPs: 
\begin{align*}
\Ln_{n_1}(\varphi) &= \frac{1}{n} L(z \mid \varphi) 
= -\frac{1}{n} \log p (z \mid \varphi) \\
&= \frac{1}{n_1} \left( \frac{n_1}{2} \log(2\pi) + \frac{1}{2} \sum_{i=1}^{n_1} (z_i - \varphi)^2 \right) 
= \frac{1}{2} \log(2\pi) + \frac{1}{2n_1} \sum_{i=1}^{n_1} (z_i - \varphi)^2, \\
\Ln(\varphi) &= \E \Ln_{n_1}(\varphi) = \frac{1}{2} \log(2\pi) + \frac{1}{2n_1} \sum_{i=1}^{n_1} \E(z_i - \varphi)^2 = \frac{1}{2} \log(2\pi) + \frac{1}{2}\{1+(\varphi^\star-\varphi)^2\}. \\
\Mn_{n_2} (\theta) &= \frac{1}{n_2} M(w \mid \theta) = - \frac{1} {n_2} \log p(w \mid \theta) \\
&= \frac{1} {n_2}  \left( \frac{n_2}{2}\log(2\pi) 
+ \frac{1}{2}\sum_{j=1}^{n_2} (w_j - (\varphi + \eta))^2 \right)
= \frac{1}{2}\log(2\pi) 
+ \frac{1}{2n_2}\sum_{j=1}^{n_2} (w_j - (\varphi + \eta))^2 \\
\Mn(\theta) &= \E \Mn_{n_2} (\theta) 
=  \frac{1}{2}\log(2\pi) + \frac{1}{2n_2}\sum_{j=1}^{n_2} \E (w_j - (\varphi + \eta))^2 \\&= \frac{1}{2} \log(2\pi) + \frac{1}{2}\left[1+\{(\varphi^\star+\eta^\star)-(\varphi+\eta)\}\right].
\end{align*}
Therefore $\Ln(\varphi, \varphi^\star) = \Ln(\varphi) - \Ln(\varphi^*) = (\varphi^\star-\varphi)^2$ and $$\Mn(\theta, \theta^\star) = \Mn(\theta) - \Mn(\theta^*) = (\varphi^\star-\varphi)^2+(\eta^\star-\eta)^2+2(\eta^\star-\eta)(\varphi^\star-\varphi).$$ 
In addition,
\begin{align*}
&D_\alpha \{ q(\varphi) \| p(\varphi) \} \\
&=  \begin{cases}
\dfrac{\alpha \log v_0}{2(\alpha-1)}  - \dfrac{\log v_\varphi}{2}  + \dfrac{ \log  \{\alpha v_0 + (1-\alpha) v_\varphi \}}{2(1-\alpha)} + \dfrac{ \alpha (\mu_\varphi-\mu_0)^2}{2 \{\alpha v_0 + (1-\alpha) v_\varphi\} } & \text{ if $\alpha \neq 1$}, \\
\left[ \log v_0 - \log v_\varphi + \{(\mu_\varphi - \mu_0)^2 + v_\varphi \}/v_0  - 1 \right] /2 & \text{ if $\alpha = 1$}.
\end{cases} 
\end{align*}
and 
\begin{align*}
& D_{\eta \mid \varphi} \{ q(\eta  \mid  \varphi) \| p(\eta \mid \varphi) \} = \\
& \begin{cases}
\frac{ \alpha \log v_b}{2(\alpha-1)} -\frac{\log v_{\eta \mid \varphi}}{2} + \frac{\log \{ \alpha v_b  + (1-\alpha)v_{\eta \mid \varphi} \}}{2(1-\alpha)} + \frac{ \alpha \mu_{\eta \mid \varphi}^2}{2\{\alpha v_b  + (1-\alpha)v_{\eta \mid \varphi}\}} & \text{ if $\alpha \neq 1$}, \\
\frac{1}{2} \left\{ \log v_b - \log v_{\eta \mid \varphi}  + (\mu_{\eta \mid \varphi}^2 + v_{\eta \mid \varphi} )/v_b - 1 \right\}  & \text{ if $\alpha = 1$}.
\end{cases}    
\end{align*}where we recall that $v_0$ and $\mu_0$ are the prior mean and variance for $p(\varphi)$, while the prior mean of $\eta\mid\varphi$ is zero, and the variance is $v_b$.
With the above, we are now ready to verify Assumptions 3 and 4 for this simple example.

Consider that $n_1\asymp n_2\asymp n$ for simplicity, and recall that we use the Gaussian variational family: $q(\varphi)\sim N(\varphi;\mu_{\varphi},v_{\varphi})$ and $q(\eta\mid\varphi)\sim N(\eta;\mu_{\eta\mid\varphi},v_{\eta\mid\varphi})$, where
$$
\mu_{\eta\mid\varphi}:=\mu_\eta + v_{\varphi, \eta} (\varphi - \mu_\varphi) / v_\varphi,\quad v_{\eta\mid\varphi}:=v_\eta- v^2_{\eta,\varphi}/v_{\varphi}.
$$

\subsubsection{$\alpha=1$: KL case}
\paragraph{Verification of Assumption \ref{ass:prior_varphi}.} Set $\rho_{1,n}(\varphi)=N(\varphi;\varphi^\star,v_0\log(n)/\sqrt{n})$. Then, we see that Assumption \ref{ass:prior_varphi}(i) is satisfied with $\epsilon_{1,n}\asymp \log(n)/\sqrt{n}$. Further, from the KL divergence calculation above, we see that (ii) is satisfied for $\epsilon_{1,n}=\log(n)/\sqrt{n}$. Choosing $\lambda_{1,n}\asymp \sqrt{n}$ and note that, for some $C$ large enough and an $n'$ large enough such that for all $n\ge n'$, 
$$
C\log (n)=C\lambda_{1,n}\epsilon_{1,n}\ge \log(n)+\epsilon_{1,n}.
$$
Lastly, note that these choices satisfy $\lambda_{1,n}/(n\epsilon_{1,n})\asymp1/\log(n)\rightarrow0$ as required. 

\paragraph{Verification of Assumption \ref{ass:prior_eta}.} Set  $\rho_{2,n}(\eta)=N(\eta;\eta^\star,v_b\log(n)/\sqrt{n})$, so that we again have $\int \Mn(\theta,\theta^\star)\dt\rho_{1,n}(\varphi)\dt\rho_{2,n}(\eta\mid\varphi)\lesssim \log(n)/\sqrt{n}$, verifying Assumption \ref{ass:prior_eta}(i) with $\epsilon_{2,n}=\log(n)/\sqrt{n}=\epsilon_{1,n}$. To verify Assumption \ref{ass:prior_eta}(iv), the pertinent assumption in this case, we instead verify the slightly weaker condition based on $\int \dt \rho_{1,n}(\varphi)\KL \{ \rho_{2,n}(\eta  \mid  \varphi) \| p(\eta \mid \varphi) \}$. In particular, using the KL divergence in this case, we see that the integral in question becomes, setting $v_{\eta,\varphi}=0$, and up to a  $\log\log(n)$ term,
\begin{flalign*}
&\int \dt \rho_{1,n}(\varphi)\KL \{ \rho_{2,n}(\eta  \mid  \varphi) \| p(\eta \mid \varphi) \}
\\&\asymp\int \dt\rho_{1,n}(\varphi)\frac{1}{2} \left\{ -\log \frac{\log(n)}{\sqrt{n}}  + ({\eta^\star }^2 + \log n/\sqrt{n} )/v_b - 1 \right\}
\\&\lesssim \log (n)+\log(n)/\sqrt{n}
\end{flalign*} Set $\lambda_{2,n}=\sqrt{n}$ and $\epsilon_{2,n}=\log(n)/\sqrt{n}$. Then, similar to the verification of Assumption \ref{ass:prior_varphi}, Assumption \ref{ass:prior_eta} (iv) is satisfied for some  $C>0$ and and $n'$ large enough so that for all $n\ge n'$, such that $
C\log (n)=C\lambda_{2,n}\epsilon_{2,n}\ge \log(n)+\epsilon_{2,n}
$, which also satisfies the condition that $\lambda_{j,n}/(n\epsilon_{k,n})\rightarrow\infty$ with $j,k\in\{1,2\}$. 

\subsubsection{$\alpha\ne1$}
All that remains to verify the satisfaction of our assumptions is to investigate the cases where $\alpha\ne1$. We only do so for Assumption \ref{ass:prior_varphi} as the arguments for verifying Assumption \ref{ass:prior_eta} are nearly identical under our choices for $\rho_{1,n}$ and $\rho_{2,n}$. 
\paragraph{Verification of Assumption \ref{ass:prior_varphi}.}
It is without loss to assume that $\wh{q}(\varphi)=N(\varphi;\varphi_n,L_n)$, for some non-random sequence $\|\varphi_n\|<\infty$, and some $L_n>0$ such that $L_n\rightarrow0$ faster than $1/\log(n)$. In the case of Assumption \ref{ass:prior_varphi}, the fraction in question becomes
\begin{flalign*}
&\frac{\KL\{\wh{q}(\varphi)\|p(\varphi)\}}{\mathcal{D}_\alpha\{\wh{q}(\varphi)\|p(\varphi)\}} =
\frac{
\displaystyle
\log\!\bigl(\tfrac{v_{0}}{L_n}\bigr)
\;+\;
\frac{(\varphi_n-\mu_{0})^{2}+L}{v_{0}}
\;-\;1}
{
\displaystyle
\frac{1}{1-\alpha}\,
\log\!\frac{\alpha v_{0}+(1-\alpha)L}{v_{0}^{\alpha}}
\;-\;\log (L_n)
\;+\;\frac{\alpha(\varphi_n-\mu_{0})^{2}}
       {\alpha v_{0}+(1-\alpha)(L)}
}.
\end{flalign*}
Write
\[
F(n)\;=\;
\frac{
\log\!\bigl(v_{0}/L_n\bigr)
\;+\;\frac{(\varphi_n-\mu_{0})^{2}+L_n}{v_{0}}
\;-\;1
}{
\frac{1}{1-\alpha}\,\log\!\frac{\alpha v_{0}+(1-\alpha)L_n}{v_{0}^{\alpha}}
\;-\;\log L
\;+\;\frac{\alpha(\varphi_n-\mu_{0})^{2}}{\alpha v_{0}+(1-\alpha)L_n}}
\,.
\]
 As \(n\to\infty\), \(L\to0\) the numerator and denominator grow like $-\log L$ 
and for for large \(n\) one has  
   \[
   F(n)
   \;=\;
   1
   \;+\;
   \frac{A - B}{-\,\log L_n \;+\;B}
   \;+\;o\!\bigl(\tfrac1{\log n}\bigr),
   \]
   where \(A\) and \(B\) are the finite constants
   \[
     A = \log v_{0} + \frac{(\varphi_n-\mu_{0})^{2}}{v_{0}}-1,
     \qquad
     B = \frac{1}{1-\alpha}\,\log\!\frac{\alpha v_{0}}{v_{0}^{\alpha}}
             +\frac{\alpha(\varphi_n-\mu_{0})^{2}}{\alpha v_{0}}.
   \]
As $n\rightarrow\infty$, \(F(n)\to1\). Further, for \(n\) large enough such that $|\log L|+B\ge 2|A-B|$, then
\[
F(n)\le 1+\tfrac12,
\]so that the assumption is satisfied with $W_n^{\alpha}=2$, for all $n$ large enough and all $\{0<\alpha<1\}\cap\{\alpha>1\}$.

\vskip 0.2in
\bibliography{ref}

\end{document}